\documentclass[journal]{IEEEtran}
\newcounter{relctr} 
\everydisplay\expandafter{\the\everydisplay\setcounter{relctr}{0}} 

\newcommand\labelrel[2]{%
  \begingroup
    \refstepcounter{relctr}%
    \stackrel{\textnormal{(\alph{relctr})}}{\mathstrut{#1}}%
    \originallabel{#2}%
  \endgroup
}
\AtBeginDocument{\let\originallabel\label} 

\usepackage{subcaption}
\usepackage{amsmath}
\usepackage{amssymb}
\usepackage{amsthm}
\usepackage{slashbox,pict2e}
\usepackage{xcolor}
\usepackage{float}
\usepackage{stfloats}
\usepackage{graphicx}
\usepackage{lineno,hyperref}
\modulolinenumbers[5]
\usepackage{epstopdf}
\usepackage{graphicx,cite,multirow,rotating}
\usepackage{setspace}
\usepackage{pgfplots} 
\usepackage{tikz}
\usepackage{enumitem}   

\usepackage[utf8]{inputenc}
\usepackage{tabularx}
\usepackage{booktabs}

\usetikzlibrary{shapes,arrows}
\usepackage{varwidth}
\usetikzlibrary{shapes.geometric, arrows}
\usetikzlibrary{positioning,arrows}

\usepackage{amsfonts}
\usepackage{yfonts}
\usepackage{psfrag}
\usepackage{pifont}
\usepackage{amsfonts}
\usepackage{bbm}
\usepackage{dsfont}

\usepackage{color}
\usepackage{amsmath,stackrel,dsfont}
\usepackage{amssymb,hyperref,stmaryrd}
\usepackage{algorithm}
\usepackage{algpseudocode}
\usepackage{pifont}
\modulolinenumbers[5]
\usepackage{amsmath}
\usepackage{epstopdf}
\usepackage{graphicx,cite,multirow,rotating}
\usepackage{setspace}
\usepackage{pgfplots} 
\usepackage{tikz}

\usetikzlibrary{shapes,arrows}
\usepackage{varwidth}
\usetikzlibrary{shapes.geometric, arrows}
\usetikzlibrary{positioning,arrows}
\usepackage{amssymb}
\usepackage{amsfonts}
\usepackage{yfonts}
\usepackage{psfrag}
\usepackage{pifont}
\usepackage{amsfonts}
\usepackage{bbm}
\usepackage{dsfont}
\usepackage{subcaption}
\usepackage{color}
\usepackage{amsmath,stackrel,dsfont}
\usepackage{amssymb,hyperref,stmaryrd}
\usepackage{algorithm}
\usepackage{algpseudocode}
\usepackage{pifont}
\usepackage{amsmath}

\usepackage{titlesec}

\newcommand{\RN}[1]{%
	\textup{\uppercase\expandafter{\romannumeral#1}}%
}

\newtheorem{Lemma}{Lemma}

\begin{document}

\title{Power Efficient Cooperative Communication within IIoT Subnetworks: Relay or RIS?}

\author{Hamid~Reza~Hashempour, \textit{Member, IEEE},
 Gilberto~Berardinelli, \textit{Senior Member, IEEE},
 Ramoni~Adeogun, \textit{Senior Member, IEEE}, and Eduard~A.~Jorswieck, \textit{Fellow, IEEE}
	\thanks{
 An earlier version of this paper was presented in part at the IEEE Joint European Conference on Networks and Communications \& 6G Summit (EuCNC/6G Summit), Antwerp, Belgium, 2024.
 
Hamid Reza Hashempour, Gilberto Berardinelli and Ramoni Adeogun’s work is 
funded by the HORIZON JU-SNS-2022-STREAM-B-01-03 6G-SHINE
project (Grant Agreement No.101095738).

Hamid Reza Hashempour, Gilberto Berardinelli and Ramoni Adeogun
are with the Department of Electronic Systems, Aalborg University, Aalborg, Denmark
		(e-mails: \{hrh, gb, ra\}@es.aau.dk).

  Eduard A. Jorswieck is with the Institute for Communication Technology, Technische Universitat Braunschweig, Germany (email: 
  e.jorswieck@tu-braunschweig.de).
	}}

\markboth{IEEE ,~Vol.~?, No.~?,  ~}%
{Shell \MakeLowercase{\textit{et al.}}: Bare Demo of IEEEtran.cls for IEEE Journals}

\maketitle

\begin{abstract}
The forthcoming sixth-generation (6G) industrial Internet-of-Things (IIoT) subnetworks are expected to support ultra-fast control communication cycles for numerous IoT devices. However, meeting the stringent requirements for low latency and high reliability poses significant challenges, particularly due to signal fading and physical obstructions. In this paper, we propose novel time division multiple access (TDMA) and frequency division multiple access (FDMA) communication protocols for cooperative transmission in IIoT subnetworks. These protocols leverage secondary access points (sAPs) as Decode-and-Forward (DF) and Amplify-and-Forward (AF) relays, enabling shorter cycle times while minimizing overall transmit power.
A classification mechanism determines whether the highest-gain link for each IoT device is a single-hop or two-hop connection, and selects the corresponding sAP. We then formulate the problem of minimizing transmit power for DF/AF relaying while adhering to the delay and maximum power constraints. In the FDMA case, an additional constraint is introduced for bandwidth allocation to IoT devices during the first and second phases of cooperative transmission. To tackle  the nonconvex problem, we employ the sequential parametric convex approximation (SPCA) method.
We extend our analysis to a system model with reconfigurable intelligent surfaces (RISs), enabling transmission through direct and RIS-assisted channels, and optimizing for a multi-RIS scenario for comparative analysis.
Simulation results show that our cooperative communication approach reduces the emitted power by up to 4.5 dB while maintaining an outage probability and a resource overflow rate below $10^{-6}$. While the RIS-based solution achieves greater power savings, the relay-based protocol outperforms RIS in terms of outage probability.
\end{abstract}


\IEEEpeerreviewmaketitle

\section{Introduction}\label{intro}
\IEEEPARstart{T}he next generation of industrial Internet-of-Things (IIoT) systems requires communication with ultra-low latency of $100 \mu s$ and ``cable-like" reliability, achieving availability levels exceeding five nines. These demanding requirements surpass the capabilities of current 5G technology, necessitating a significant transformation in network infrastructure that integrates intelligence and decision-making at the network edge \cite{Adeogun2020, Gilberto}. To address these challenges, researchers and industry professionals are actively exploring short-range, low-power in-X subnetworks designed to replace wired control systems in environments such as robotics and production facilities. The objective of these initiatives is to develop a cost-effective wireless solution that meets the stringent needs of industrial automation \cite{IIoT1,IIoT2,Wang2024,Chen2024}.
Despite the limited propagation range of these networks, achieving the required latency and reliability levels presents challenges due to signal blockage and fading, often caused by metallic obstacles. Relaying has emerged as an effective strategy for reducing fading and enhancing transmission reliability by utilizing spatial diversity, thereby supporting higher data rates and lower latency \cite{URLLC1,Relay1}.

To overcome the effects of blockage and facilitate ultra-short transmission cycles in IIoT subnetworks, incorporating link diversity is essential. Numerous studies have explored the use of relays to enable ultra-reliable low-latency communication (URLLC) networks, as discussed in \cite{Relay1,URLLC1,Khosravirad,Cheng,Ranjha,Joint,Relay-select}. For instance, \cite{URLLC1} offers a comprehensive overview of the performance benefits associated with relaying in 5G URLLC transmissions compared to direct transmission methods. Similarly, \cite{Khosravirad} introduces a wireless communication protocol tailored for industrial networks, which enhances the performance of nodes with weak links through cooperative transmission while using direct transmission for stronger channels. However, this work does not address power control issues or provide an algorithm for classifying  IoT devices\footnote{For simplicity, `device' will be used to refer to IoT devices throughout the remainder of the paper.} into single-hop and two-hop categories. An alternative approach, presented in \cite{Cheng}, proposes an algorithm aimed at minimizing transmit power in smart factories, focusing on throughput and reliability constraints rather than low latency. Additionally, \cite{Ranjha} investigates the application of mobile airborne relays to enable short control packet transmissions between controllers and multiple mobile robots, although this optimization emphasizes minimizing overall decoding error probability rather than reducing power emissions.
Furthermore, the authors in \cite{Joint} propose a method for jointly selecting an optimal relay vehicle to enhance channel conditions and control power, ensuring communication link reliability for effective broadcast communication in platooning scenarios. The issue of relay selection has also been examined in the context of RF energy harvesting \cite{Relay-select}. 
Recent studies have addressed the challenges of power and frequency allocation for 6G industrial subnetworks in \cite{Power-Alloc} and \cite{Freq-Alloc}, focusing on these aspects independently. However, these works have not considered the potential benefits of cooperative communication for improving URLLC performance.

While previous research has examined relay-assisted URLLC transmission, including aspects such as relay selection and power control, a notable gap remains in the literature regarding a holistic approach that integrates transmission protocol design, relay selection, and power control, specifically tailored to the stringent requirements of IIoT subnetworks. These subnetworks demand ultra-low latency of $100 \mu s$ and service availability exceeding five nines. Furthermore, existing studies have largely overlooked the critical issue of minimizing power emissions, which is essential for extending battery life and reducing interference with neighboring subnetworks \cite{Gilberto}.

Reconfigurable Intelligent Surfaces (RIS), which can modify wireless propagation environments, have gained recognition as a promising method for boosting link performance in difficult scenarios, while also being cost-efficient and energy-saving \cite{Mu,Liu}. In the context of 6G subnetworks, RIS offers an alternative approach to enhance wireless control by deploying multiple surfaces to reduce signal blockage and improve URLLC performance \cite{Zhang}. As demonstrated in \cite{Zhang}, distributed RIS deployments outperform colocated RIS configurations. However, there is a lack of direct performance comparisons between systems that use multiple access points (APs) as relays and those that employ RISs.

\subsection{Motivations and Contributions}
As discussed in the introduction, the URLLC requirements of IIoT subnetworks, particularly the need for ultra-short communication cycles, demand innovative cooperative transmission protocols.
Existing literature lacks a comprehensive approach integrating transmission protocol design, relay selection, and power control for these demanding settings.
In our previous work \cite{EuCNC}, we introduced a novel time division multiple access (TDMA) protocol to address these requirements. Building upon that foundation, this paper expands the approach by incorporating both amplify-and-forward (AF) and DF relaying. Additionally, we propose a frequency division multiple access (FDMA) protocol designed for cooperative communication. Our aim is to enhance reliability while minimizing transmit power, addressing the unique challenges encountered in IIoT environments.
The proposed solution leverages secondary APs (sAPs) to relay packets from devices experiencing poor propagation conditions to a primary AP (pAP), which includes embedded controller functionalities. This configuration ensures robust communication even in the presence of signal blockages and fading—issues commonly encountered in industrial settings.
We present a thorough analysis of the framework by examining both DF and AF relaying schemes in the context of TDMA and FDMA. 

Another gap in the literature is the lack of direct performance comparisons between systems using multiple APs as relays and those employing RISs. While studies \cite{RIS-or-relay1,RIS-or-relay2,RIS-or-relay3} have addressed comparisons between RISs and relays, none have focused on minimizing transmit power under the stringent URLLC requirements of IIoT environments. Our work bridges this gap by addressing both power minimization and the use of multiple relays or RISs, setting it apart from existing studies.
To address this, we extend our analysis to an equivalent system model where sAPs are replaced with RISs, facilitating transmission through both direct device-to-pAP channels and device-RIS-pAP channels. This analysis also includes addressing the challenges associated with CSI estimation in the RIS-based setup.
The main contributions of this work are as follows:
\begin{itemize}
\item We propose novel transmission protocols for both TDMA and FDMA configurations in DF/AF relaying systems designed for URLLC applications. In these protocols, devices are dynamically scheduled into single-hop or two-hop transmission modes, utilizing the most suitable sAP based on real-time channel state information (CSI), while adhering to the strict timing constraints (of IIoT subnetworks).

\item We extend our analysis to an equivalent system model where sAPs are replaced with RISs. In this scenario, transmission occurs through both the direct device-to-pAP channel and the combined device-RIS-pAP channel. We also address the challenge of CSI estimation for this setup.

\item We formulate and address the critical problem of minimizing total transmit power while ensuring adherence to stringent communication cycle time requirements. This optimization problem is developed for both DF/AF relaying systems and distributed RIS configurations, providing a unified framework for power efficiency across different setups.

\item To tackle the inherent non-convexity of this optimization problem, we utilize a sequential parametric convex approximation (SPCA) method, which is well-suited to complex, non-linear optimization scenarios. In the FDMA cases, we further introduce a sequential decoupling approach to optimize bandwidth allocation by isolating specific parameters. This dual approach significantly enhances overall system performance by efficiently managing power and bandwidth resources.

\item We provide extensive simulation results to characterize the performance of the proposed protocols, particularly in terms of power efficiency. Additionally, we conduct a comparative analysis of our proposed cooperative method against a scenario with multiple RISs, focusing on power savings and reliability.
\end{itemize}

To the best of our knowledge, no prior work has integrated transmission protocols for DF/AF relaying in both TDMA and FDMA configurations under stringent timing constraints, while comparing them with RIS scenario and addressing the joint classification of devices and power minimization strategies within IIoT  subnetworks.

\subsection{Organization and Notation}
The remainder of this paper is organized as follows: Section~\ref{Sys_Model} describes the system model. Section~\ref{Comm_protocol} and Section~\ref{RIS-model} detail the proposed communication protocols for relay-assisted and RIS-assisted networks, respectively. In Section~\ref{opt} and \ref{RIS-opt}, we provide a comprehensive explanation of the power minimization algorithms tailored for relay, and RIS-aided scenarios, respectively. Section~\ref{Simulat} presents simulation results that demonstrate the effectiveness of our approach. Finally, the conclusions are summarized in Section~\ref{conc}.

\textit{Notation}: We use bold lowercase letters for vectors and bold uppercase letters for matrices. The notation $(\cdot)^T$ and $(\cdot)^H$  denote the transpose operator and the conjugate transpose operator, respectively.  $\triangleq$ denotes a definition. $\mathbb{R}^{N \times 1}$ and $\mathbb{C}^{N \times 1}$ denote the sets of $N$-dimensional real and complex vectors, respectively.  $\mathbb{C}^{M \times N}$ stands for the set of $M \times N$ complex matrices.   $\mathrm{diag}\{\cdot\}$ constructs a diagonal matrix from its vector argument.

Table \ref{Table-1} presents the main parameters and variables associated with this study in order to enhance the readability of the paper. 
\begin{table}
	\small
	\renewcommand{\arraystretch}{1.3}
	\caption{Key notations used in this paper.}
	\centering
	\label{Table-1}
	\resizebox{\columnwidth}{!}{
		\begin{tabular}{|c|p{60mm}|}
			\hline
			$\mathbf{Notation}$  &  $\mathbf{Definition}$ \\
			\hline
             $\mathcal{N} /N /n$  &  Set/number/index of devices \\		\hline
			$\mathcal{K} /K /k$  &  Set/number/index of sAPs (or RISs) in the uplink transmission \\		\hline
			$\mathcal{N}_{1h}/ \mathcal{N}_{2h}$ &  The set of devices scheduled for single-hop/two-hop transmission \\ \hline
 $D_n/\mathcal{D}$  & The strongest sAP for the $n$-th device/the set of all selected sAPs
\\ \hline   
		$h^d_n /h^s_{n,D_n}$	  & 	The  channel  of $n$-th devices to the pAP/sAP
   \\		\hline
    $h^a_{D_n}$ & The channel vector between the pAP and $D_n$th sAP
   \\		\hline
$\mathbf{h}_{n,k}/ \mathbf{h}^r_k$ & Channel responses from device $n$ to
 RIS $k$/ from RIS $k$ to the pAP
    \\		\hline $P_n$, $P_k^s$
    & The transmit power of $n$-th/$k$-th device/sAP (W).
    \\		\hline $B_n$; $W$
    & Packet size of $n$-th device (Bit); available bandwidth (Hz).
    \\		\hline $T$/$L$
    & Uplink timeslot (s)/ number of training symbols.
    \\		\hline
     $\beta_{n}$/$\beta^s_{n}$ & The bandwidth allocation parameter
for device $n$ in first/second phase of transmission
\\		\hline
			$\mathcal{J}_k /J_k /j$ &  Set/number/index of $k$-th RIS elements  \\ 	\hline
			$\boldsymbol{\Phi}_k/\boldsymbol{\phi}_k$ &  The diagoal matrix/vector of the phase shift of the $k$-th RIS    \\ 		\hline
        $\sigma_0$; $\sigma_e$  & AWGN noise power (W); channel estimation error power (W).
			    \\ 		\hline
		$\theta$	& Discount factor for the rate of data transmission
			\\		\hline
		\end{tabular}}
	\end{table}

\section{Setting Up the Scene}\label{Sys_Model}
We explore two cooperative transmission scenarios within an IIoT subnetwork: one using multiple APs and the other using multiple RISs to enhance transmission. Each model is detailed below.

\subsection{Multiple APs}
In the first scenario, an IIoT subnetwork consists of $N$ devices wirelessly connected to $K$ sAPs and one pAP. The pAP issues control commands to the devices, while the $K$ sAPs provide radio communication support. The sets of all sAPs and devices are represented by $\mathcal{K}$ and $\mathcal{N}$, respectively. All devices and APs are equipped with a single antenna. The system model for uplink (UL) transmission with multiple APs is illustrated in Fig.~\ref{fig1a}.
In this setup, devices communicate either directly with the pAP in a single-hop mode or relay their messages through an sAP to enable a two-hop cooperative transmission. Devices using single-hop transmission are denoted by $\mathcal{N}_{1h}$, while those in two-hop mode are denoted by $\mathcal{N}_{2h}$.
An example subnetwork with five devices and three APs is shown in Fig.~\ref{fig1a}, where devices D and E use single-hop transmission, while devices A, B, and C use two-hop cooperative mode. In the two-hop case, devices transmit their packets directly to the pAP in the first phase and via an sAP in the second, allowing the pAP to combine signals received directly with those relayed by the sAP.
We consider both DF and AF relaying methods. Each relaying strategy includes two configurations: (1) a TDMA setup with devices assigned time slots to avoid intra-cell interference and (2) an FDMA setup with orthogonal subchannels allocated within a common time slot.
In both configurations, we focus on UL transmissions, where devices send packets of $B_n$ bits over $W$ Hz bandwidth to the APs. All packets must reach the primary AP within a $T$-second time slot.
\subsection{Multiple RISs}
The second scenario, illustrated in Fig.~\ref{fig1}, replaces the sAPs with RISs. Here, we have one pAP and $K$ distributed RISs, each with $J_k$ reflecting elements. Devices communicate with the pAP through both direct links and cascaded device-RIS-pAP links.
We assume that the pAP has CSI estimates for all communication links, including device-to-pAP/sAP/RIS links and sAP/RIS-to-pAP links. These CSI estimates are obtained through pilot transmission during a training phase, as outlined in \cite{Khosravirad}. For static devices, the CSI remains relatively stable, allowing less frequent training. Interference from neighboring subnetworks is assumed to be negligible due to effective network-level radio resource management.
The pAP scheduler uses the CSI to determine whether devices should operate in single-hop or two-hop transmission modes. It also optimizes rate, time and frequency resource allocation, and transmit power by solving an optimization problem.
\begin{figure} 
	\centering
	\subfloat[\label{fig1a}]{%
		\includegraphics[width=0.92\linewidth]{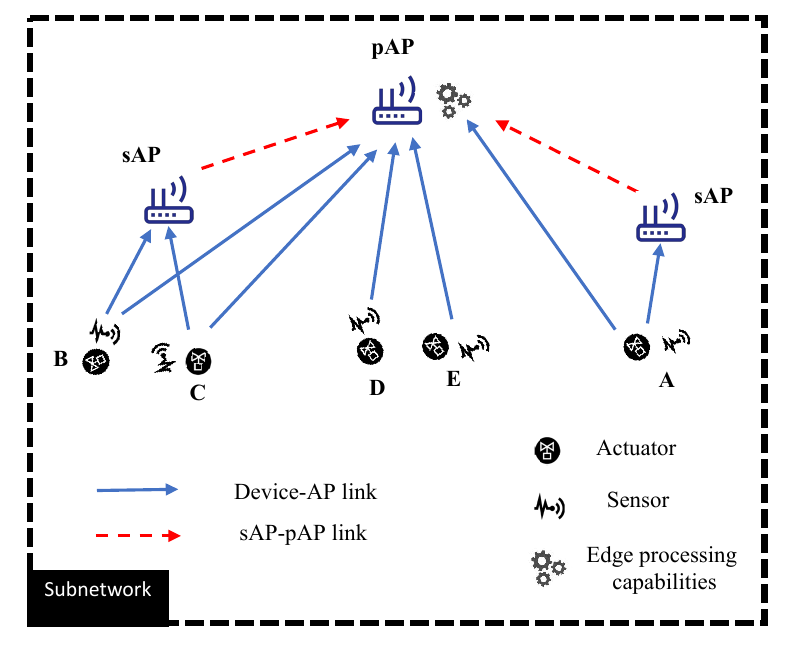}}
	\hfill
	\subfloat[\label{fig1}]{%
		\includegraphics[width=0.92\linewidth]{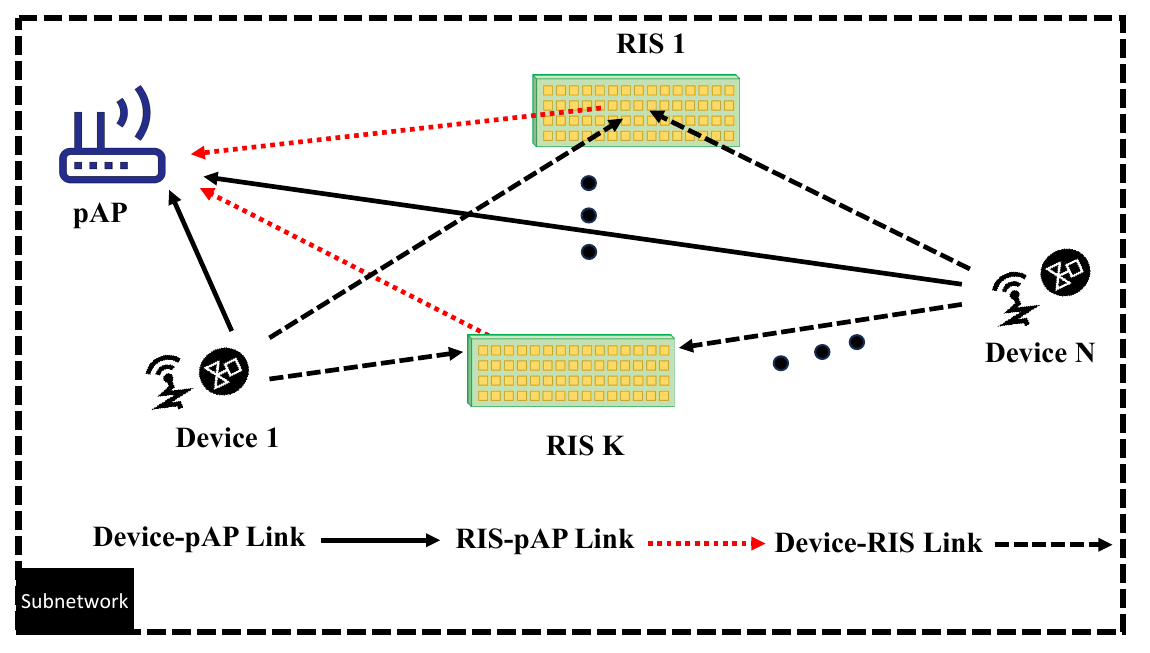}}
	\caption{(\protect\subref{fig1a}) System model for transmission with multiple APs and (\protect\subref{fig1}) multiple RISs in the subnetwork.}
		\label{fig-sys}
\end{figure}
To clarify the communication process and the role of each component, we provide a step-by-step overview (Fig.~\ref{process_fig}):
\begin{itemize}
\item \textbf{Channel Estimation:} Devices and sAPs send pilot sequences, enabling the pAP to estimate channel quality. Based on these estimates, the pAP assigns devices to single-hop or two-hop transmission modes to balance efficiency and reliability.

\item \textbf{Resource Management:} The pAP manages resources including time slots, bandwidth, and power allocation for UL transmission. Resource allocation is dynamically adjusted for single-hop and two-hop devices to optimize power efficiency and latency.

\item \textbf{UL Transmission:} Devices initiate data transmission based on the pAP's scheduling. Devices assigned to two-hop transmission first send data to both an sAP and the pAP, while single-hop devices transmit  directly to the pAP.

\item \textbf{Relaying:} Selected sAPs relay signals using AF or DF methods. During the second phase of cooperative transmission, sAPs forward the received signals to the pAP, completing the two-hop transmission process.
\end{itemize}
\begin{figure}
	\centering
		\includegraphics[width=0.92\linewidth]{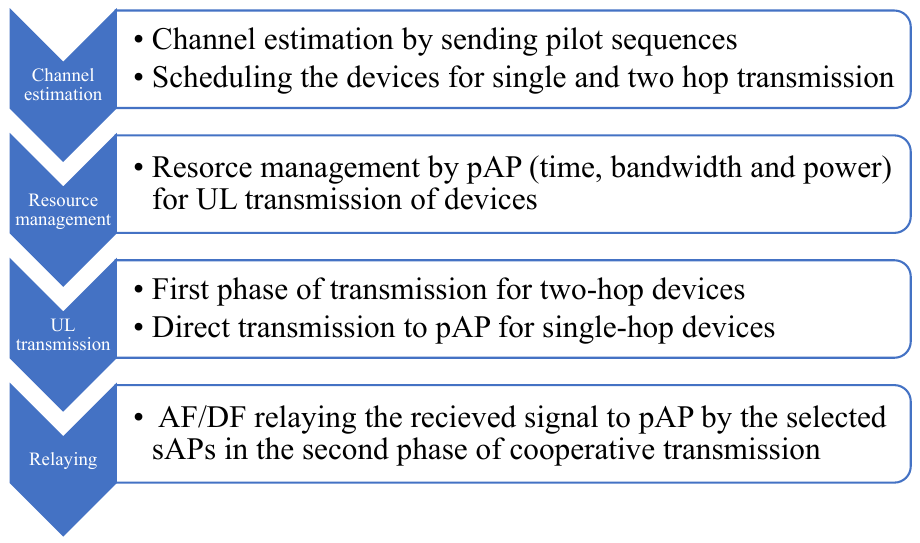}
	\caption{Overview of the communication process.}
	\label{process_fig}
 \end{figure}
\section{Proposed Communication Protocols for Relay-assisted Network}\label{Comm_protocol}
In this section we elaborate our communication protocols in TDMA and FDMA cases, respectively.
It is our hypothesis that communication reliability in a subnetwork can be enhanced by the presence of sAPs and their forwarding capabilities. Exploiting the presence of the sAPs requires a tailored communication protocol. 
    
Our proposed TDMA protocol is shown in Fig.~\ref{fig2}. The UL timeslot is divided into three variable-duration sub-slots: the first phase of two-hop transmissions, single-hop transmissions, and the second phase of two-hop transmissions. Devices in the set $\mathcal{N}_{2h}$ transmit during the first sub-slot, with their signals received by both the pAP and the sAP. Each device in $\mathcal{N}_{2h}$ is served by the sAP with the best channel conditions.
In the second phase sub-slot, the sAP forwards the received messages to the pAP, acting as a DF/AF relay. The pAP combines the energy from the first-phase signal with that forwarded by the sAP before decoding. Single-hop transmissions are only received by the pAP.
To allow sufficient processing time at the sAPs, devices in two-hop mode are scheduled with a gap between receiving and retransmitting packets. This ensures efficient data processing without extending the timeslot duration.
It is worth observing that, given the need of accommodating the two phases of the cooperative operation in the UL slot, transmission intervals of devices in $\mathcal{N}_{2h}$ are necessarily shorter than in the basic TDMA protocol, and therefore a higher transmission rate is needed for transmitting their packet. 

\begin{figure} 
	\centering
	\subfloat[\label{fig2}]{%
		\includegraphics[width=0.920\linewidth]{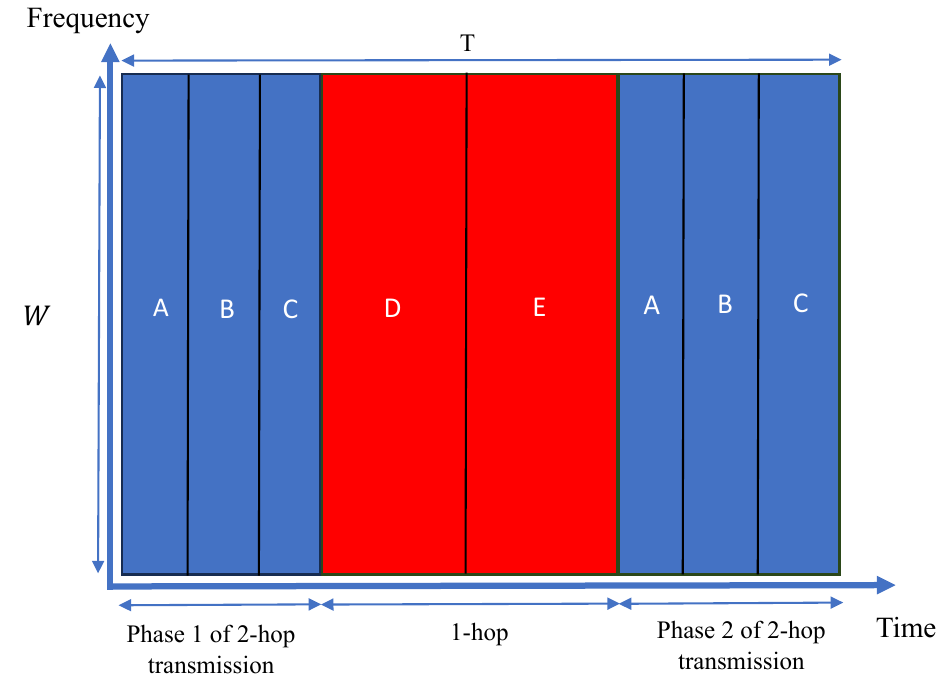}}
	\hfill
	\subfloat[\label{fig4}]{%
		\includegraphics[width=0.920\linewidth]{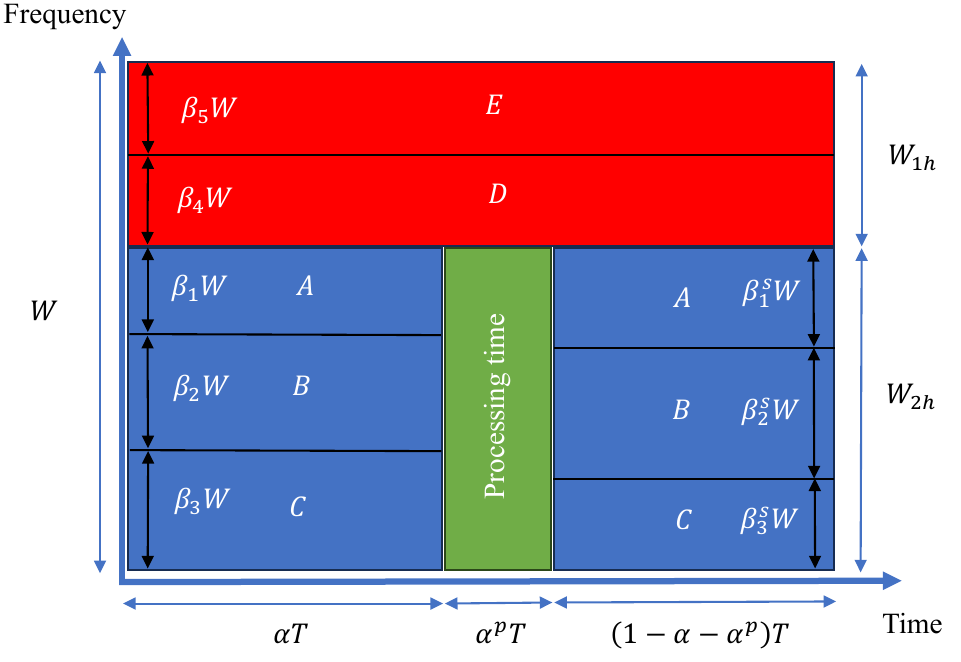}}
	\caption{(\protect\subref{fig2}) Proposed cooperative TDMA  and (\protect\subref{fig4})  FDMA protocols for relay-assisted network.}
		\label{fig6}
\end{figure}

Conversely, in the case of FDMA, all devices use the same time slot for transmission, but the bandwidth is divided among the devices to prevent interference.
In Fig.\ref{fig4}, the proposed cooperative FDMA protocol is illustrated. It can be observed that for two-hop transmission in the FDMA case, a processing time is added to account for DF processing at the sAPs.
$\beta_{n}$ and $\beta^s_{n}$ in Fig.\ref{fig4} denote the fraction of bandwidth allocated
to device $n$ in the first and second phase of cooperative transmission, respectively. Due to limited bandwidth, we have $\sum_{n \in \mathcal{N}} \beta_{n} = 1$
and $\sum_{n \in \mathcal{N}_{2h}}(\beta^s_{n} -\beta_{n})\leq  0 $. It is worth noting that in TDMA, since the total bandwidth is used by all devices, we set $\beta_{n}=1$ and $\beta^s_{n}=1$.

In the following, we present a model for channel estimation error, and the signal model of single-hop and two-hop DF relaying. 
For the signal model, we consider the following notation. The vector of channel responses between the pAP and each of $K$ sAPs is denoted by $h^a_{k}$. The channel of pAP link to each of $N$ devices is represented by $h^d_n$. The channel from the device $n$ to $k$-th sAP is $h^s_{n,k}$.

\subsection{Channel Modeling and Estimation in Relay Scenario}\label{ch-est}
Given the small size of subnetworks \cite{Adeogun2020, Gilberto} and the close proximity of devices to their associated APs, line-of-sight (LOS) conditions are highly probable in dense small cell network environments. This suggests that Rician fading channel models are more suitable than Rayleigh models for accurately representing multipath channel characteristics in such deployments\cite{Jafari}. For this reason, we adopt the Rician fading model in this study, utilizing the 3GPP channel model  for the indoor factory scenario as outlined in \cite{3GPP}. Additionally, we account for the impact of imperfect CSI (I-CSI) by incorporating channel estimation based on pilot training sequences for each device. Suppose each device uses $L$ training symbols, with a duration of $T_p = L T_s$, where $T_s = 1/W$ is the symbol period and $W$ is the total transmission bandwidth. The pilot training time for all devices is subtracted from the total time slot $T$, resulting in $T' \triangleq T - N L T_s$, where $T'$ represents the remaining time available for UL transmission.
Assuming that recursive minimum mean-square-error (MMSE) channel estimation is used, the true Rician fading channel gain $h$ can be written as $h = \hat{h} + \epsilon$, where $\epsilon \sim \mathcal{CN}(0,\sigma_e(L))$, $\hat{h} \sim \mathcal{CN}(0,1-\sigma_e(L))$, \cite{Khosravirad} and
\begin{align}\label{MMSE}
\sigma_e(L) = \dfrac{1}{1+L\cdot \mathrm{SNR}},
\end{align}
where $\sigma_e(L)$ denotes the variance of the channel estimation
error \cite{Khosravirad}.
To reduce the impact of outages caused by inaccuracies in channel estimation, the transmitter can lower the data transmission rate by applying a discount factor $\theta$, where $0 < \theta < 1$.

\subsection{Achievable Rate for Single-hop Transmission}
For simplicity of notation let us consider the perfect CSI (P-CSI) case, i.e. the pAP
scheduler has perfect knowledge of the channel responses of
all the links for making  its decision. In the I-CSI case discussed in the previous section, the estimated channels ($\hat{h}$) are used in place of the perfect channels ($h$), and the achievable rate is adjusted by a factor $\theta$ to account for the uncertainty introduced by the channel estimation error ($\epsilon$).

The Signal to Noise Ratio (SNR) of the $n$-th device in the direct link to pAP is given by:
\begin{align} 
g_{n}^d&=\dfrac{P_n|h^d_n|^2}{\beta_n\sigma_0},\ \ \forall  n \in \mathcal{N}_{1h}, 
\end{align}
where the superscript $d$ stands for the direct link, $\sigma_0$ denotes power of
the additive white Gaussian noise (AWGN), and $P_n$ is the transmit power of device $n$. 
The achievable information rate of the $n$-th device in the direct link to pAP can be written as:
\begin{align}
r_{n}^d &= W \beta_{n} \mathrm{log}_2 \left( 1+ g_{n}^d  \right), \ \ \forall  n \in \mathcal{N}_{1h}. 
\end{align}
Assuming time division multiplexing only over a bandwidth $W$, a packet of $B_n$ bits for the $n$-th device can be transmitted in a time $t_{n}^d = \dfrac{B_n}{ r_{n}^d }$. The packets of all single-hop devices can be then transmitted in a TDMA fashion resulting in a total transmission time: 
\begin{align}
 T_{1h}^T &= \sum_{n \in \mathcal{N}_{1h}} t_{n}^d.
\end{align}
In FDMA, all devices transmit simultaneously, thus the total transmission time of all devices is 
\begin{align}
 T_{1h}^F &= \max_{  n \in \mathcal{N}_{1h} } \ t_{n}^d. 
\end{align}

\subsection{Achievable Rate for DF Relaying}
In the DF relaying model, each device’s transmission is first received by an sAP, which decodes and then forwards it to the pAP. We propose a relaying strategy called the `$1$ of $K$' method, where each device is assigned to the sAP with the strongest channel conditions (i.e., highest channel gain) during both phases of transmission. This strategy enhances communication reliability and coverage. The sAP with the strongest signal for the $n$-th device is denoted by $D_n$, where $D_n \in \mathcal{D} \subseteq \mathcal{K}$. The SNR and achievable rate for the $n$-th device at the $D_n$-th sAP are specified accordingly.
\begin{align}
g^s_{n,D_n}&=\dfrac{P_n|h^s_{n,D_n}|^2}{\beta_{n}\sigma_0},\ \ \forall  n \in \mathcal{N}_{2h}, \\
r_{n,D_n}^{(1)} &= W \beta_{n} \mathrm{log}_2 \left( 1+ g^s_{n,D_n}  \right), \ \ \forall  n \in \mathcal{N}_{2h},
\end{align}
where the superscript ${(1)}$ indicates the first phase in the two-hop cooperative transmission.
Assuming the sAP re-encodes and forwards the received signal to the pAP, the SNR and the achievable rate from the $D_n$-th sAP to the pAP are given by the following expressions:
\begin{align}\label{eq7}
g^a_{D_n}&=\dfrac{P_{D_n}^s|h^a_{D_n}|^2}{\beta^s_{n}\sigma_0},\ \forall  n \in \mathcal{N}_{2h},  \\
r^{(2)}_{n} &=  W \beta^s_{n} \mathrm{log}_2 \left( 1+ g^a_{D_n} + g_{n}^d \right),\ \forall  n \in \mathcal{N}_{2h},
\label{eq7}
\end{align}
where the superscript ${(2)}$ indicates the second phase in the two-hop cooperative transmission. 
Additionally, $P_{D_n}^s$ represents the transmit power of the $D_n$-th sAP. It is important to note that \eqref{eq7} accounts for coherent combining of the two transmission phases at the pAP receiver. For DF relaying, the signal sent by device $n$ must be correctly decoded by the strongest sAP and re-encoded into a new message. The total over-the-air time required for device $n$ to successfully transmit a packet of $B_n$ bits is then given by:
\begin{align}\label{tn}
t_n^{(2h)}&=\dfrac{B_n}{  \displaystyle r_{n,D_n}^{(1)} }+\dfrac{B_n}{  r_n^{(2)}},\ \forall  n \in \mathcal{N}_{2h}.
\end{align}
We denote respectively the total time of transmission for all devices in the first and second phase of two-hop method in TDMA case as: 
\begin{align}\label{eq10}
	\allowdisplaybreaks
T_{2h}^{(1)} & = \sum_{n \in \mathcal{N}_{2h}} \dfrac{B_n}{  r_{n,D_n}^{(1)} }, \\
T_{2h}^{(2)} & = \sum_{n \in \mathcal{N}_{2h}} \dfrac{B_n}{  r_n^{(2)}}.
\end{align}
The total time of transmission for all single-hop and two-hop devices in TDMA mode is calculated by the summation of $T_{1h}^{T}$, $T_{2h}^{(1)}$, and $T_{2h}^{(1)}$. In addition,
for FDMA case the transmission time for each device in two-hop category is calculated individually from \eqref{tn}, thus we define
\begin{align}
T_{2h}^{F}= \max_{  n \in \mathcal{N}_{2h} } (t_n^{(2h)}+t_n^{p}),
\end{align}
where $t_n^{p}$ is the processing time of device $n$.
To summarize we define the transmission time of our proposed method as:
\begin{align}
T^{DF}=
\begin{cases}
T_{1h}^T +T_{2h}^{(1)} + T_{2h}^{(2)}, & \text{for TDMA mode,}\\
\max \ (T_{1h}^F,T_{2h}^{F}) , & \text{for FDMA mode.}
\end{cases}
\end{align}

Every device expects independent $B_n$ bits
of data to be delivered in $T$ seconds over a bandwidth of $W$ Hz.  The total transmission time must not exceed the predefined time $T$, i.e., $T^{DF} \leq T$. Otherwise, a time overflow (resource overflow) occurs. We utilize the \textit{overflow rate}  to quantify the occurrence of events where $T^{DF} > T$.
Note that time overflow is the only source of errors in the P-CSI case. However, in the I-CSI case, we also consider the probability of outage. The outage probability for single-hop and two-hop devices can be calculated, respectively, as follows:
\begin{align}
P_{\text{out}}^{1h} &= \mathrm{Pr}\left(\exists i \in \mathcal{N}_{1h}, \theta \cdot \hat{r}^d_i > r^d_i\right), \nonumber \\
P_{\text{out}}^{2h} &= \mathrm{Pr}\left(\exists i \in \mathcal{N}_{2h}, \theta \cdot \hat{r}^{(1)}_{i,D_i} > r^{(1)}_{i,D_i} \cup \theta \cdot \hat{r}^{(2)}_{i} > r^{(2)}_{i}\right)
\end{align}
where $P_{\text{out}}^{1h}$ and $P_{\text{out}}^{2h}$ represent the outage probability for single-hop and two-hop communication, respectively. $\hat{r}^d_i$, $\hat{r}^{(1)}_{i,D_i}$ and $\hat{r}^{(2)}_{i}$ denote the estimated rates for the direct link and the first and second cooperative links, respectively.
 
\subsection{Achievable rate for AF Relaying}
In this section we consider the standard AF protocol. 
In AF transmission, the relay only amplifies and retransmits the received signals without any decoding.
The received signal at the destination through the direct link is denoted as:
\begin{align}
y_{n,p} &= \sqrt{P_n}h^d_n x_n + z_0, \ \forall  n \in \mathcal{N}.
\end{align}
The received signals in source-relay, and relay-destination links, are respectively as:
\begin{align}
y_{n,D_n} &= \sqrt{P_n}h^s_{n,D_n} x_n + z_1, \ \forall  n \in \mathcal{N}_{2h},  \\ 
y_{D_n,p} &= \mu_{n} \sqrt{P^s_{D_n}} h^a_{D_n} y_{n,D_n}+ z_2,\ \forall  n \in \mathcal{N}_{2h},  
\end{align}
where $x_n$, $z_1$,  $z_2$, and $\mu_n$ are the transmitted signal,  the additive white Gaussian noise (AWGN) noise at the relay and pAP, and the amplification factor, respectively.

Assuming that the sAPs replicate sequentially the signal received by each device and energy combining is performed at the receiver, it can be shown that the SNR  can be  written as \cite{AF-ref}:
\begin{align}
g_{n,D_n}^{\mathrm{AF}}&=\dfrac{P_{n} P_{D_n}\mu_{n}^2 |h^a_{D_n}|^2 |h^s_{n,D_n}|^2}{(1+\mu_{n}^2 P_{D_n} |h^a_{D_n}|^2)\beta_{n}\sigma_0},\  \forall  n \in \mathcal{N}_{2h}, 
\label{12}
\end{align}
It is worth noting that in AF relaying, the bandwidth remains unchanged during the relaying stage. Therefore, $\beta_{n}$ is used as the bandwidth allocation parameter in both phases.
In order to have a constant power constraint we set $\mu_{n} = \sqrt{1/(P_{n}|h^s_{n,D_n}|^2+\beta_{n}\sigma_0)}$. Thus, $\forall  n \in \mathcal{N}_{2h}$ the SNR and the achievable rate are respectively obtained as:
\begin{align}
g_{n,D_n}^{\mathrm{AF}}&=\dfrac{P_{n}P_{D_n} |h^a_{D_n}|^2 |h^s_{n,D_n}|^2}{(P_{n}|h^s_{n,D_n}|^2+P_{D_n}|  h^a_{D_n}|^2 +\beta_{n}\sigma_0)\beta_{n}\sigma_0},
\nonumber \\ 
r_n^{\mathrm{AF}} &= \dfrac{W \beta_n}{2} \mathrm{log}_2 \left(1+ g_{n}^d + g_{n,D_n}^{\mathrm{AF}}\right). \label{rate-AF}
\end{align}
In \eqref{rate-AF}, the factor of $1/2$ in the data rate arises because two consecutive time slots of the same duration are allocated for each packet.
Then, the total time of transmission in AF case then can be calculated as:
\begin{align}
T^{AF}=
\begin{cases}
T_{1h}^T + \displaystyle\sum_{n \in \mathcal{N}_{2h}} \dfrac{B_n}{  r_n^{\mathrm{AF}}}, & \text{for TDMA mode,}\\
\max \ (T_{1h}^F,\displaystyle\max_{  n \in \mathcal{N}_{2h} } \dfrac{B_n}{  r_n^{\mathrm{AF}}} ), & \text{for FDMA mode.}
\end{cases}
\end{align}
The overflow rate for AF relaying is calculated similarly to that of DF relaying. The outage probability for AF relaying is given by: 
\begin{align}
P_{\text{out}}^{AF} &= \mathrm{Pr}\left(\exists i \in \mathcal{N}_{1h}, \theta \cdot \hat{r}^{\mathrm{AF}}_i > r^{\mathrm{AF}}_i\right).
\end{align}

\section{Proposed Communication Protocol for Multiple RIS-assisted Network}\label{RIS-model}
In the previous section, we considered cooperative communication using multiple sAPs to enhance URLLC transmission. Here, we substitute each sAP with an RIS, keeping the same configuration, and elaborate on the UL transmission scenario.
Therefore, consider an RIS-assisted UL channel, as depicted in Fig. \ref{fig1}, which includes one pAP, a set $\mathcal{N}$ of $N$ devices, and a set $\mathcal{K}$ of $K$ RISs. Each device and the pAP are equipped with a single antenna. Each RIS, indexed by $k \in \mathcal{K}$, contains $J_k$ reflecting elements and is configured to facilitate communication between the pAP and the devices.
\begin{figure}
	\centering	\includegraphics[width=0.920\linewidth]{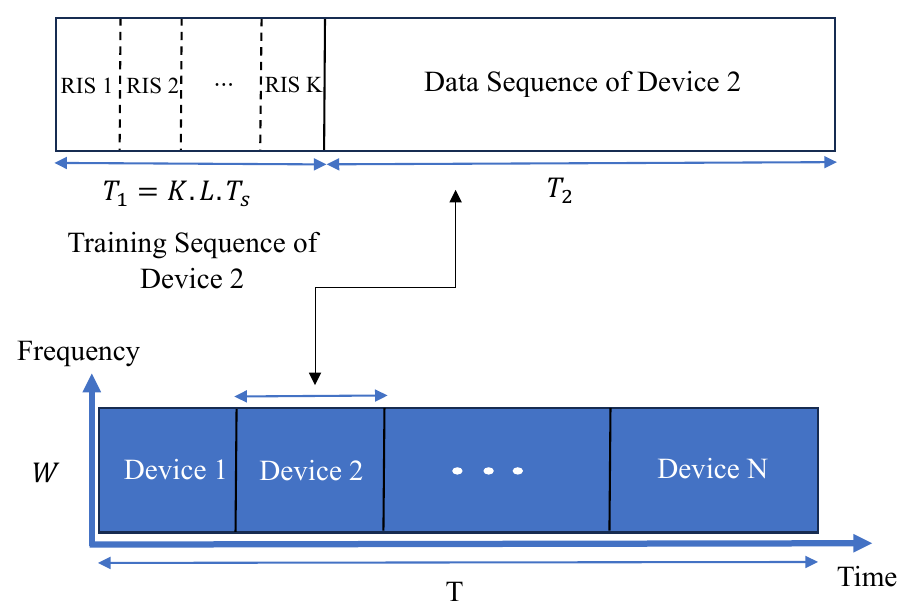}
	\caption{Protocol for UL transmission and channel estimation with multiple RISs.}
	\label{RIS_training}
\end{figure}

We assume that each RIS functions as an ideal passive element, implying it does not introduce frequency selectivity. Consequently, our model is designed to focus exclusively on TDMA transmission\footnote{It is important to note that FDMA significantly degrades performance compared to TDMA, as shown in \cite{Multi-RIS}. Moreover, while \cite{Multi-RIS} introduces a method for joint optimization in a downlink scenario, this falls outside the scope of our paper. Due to space constraints, we do not cover it here but plan to address it in future research.}.
The transmitted signal from device $n$
is $\sqrt{P_n} x_{n}$.
The phase shift
matrix of RIS $k$ can be optimized through a diagonal matrix
    $\boldsymbol{\Phi}_k = \mathrm{diag}(e^{j\phi_{k1}},\cdots,e^{j\phi_{kJ_k}})\in \mathbb{C}^{J_k \times J_k} $
with $\phi_{kj} \in [0,2\pi]$, $k \in \mathcal{K}$, and $j \in \mathcal{J}_k= \{1,2,\cdots, J_k \} $, where $\boldsymbol{\Phi}_k$ captures the
effective phase shifts applied by all reflecting elements of RIS $k$. Thus, the received signal of $n$-th device at the destination through the direct link and multiple RISs is denoted as:
\begin{align}\label{ris-eq}
y_{n} &= \sqrt{P_n}\left((h^{d}_{n})^* + \sum_{k=1}^{K} \mathbf{h}^H_{n,k} \boldsymbol{\Phi}_k \mathbf{h}^r_k  
\right)  x_n + z_0, \ \forall  n \in \mathcal{N}
\end{align}
where $h^{d}_{n}$, $\mathbf{h}_{n,k} \in \mathbb{C}^{J_k}$, and $\mathbf{h}^r_k \in \mathbb{C}^{J_k}$ , respectively,
denote the channel responses from device $n$ to the pAP, from device $n$ to
 RIS $k$, and from RIS $k$ to the pAP, and $z_0$
is the additive white Gaussian noise.
Due to the fact that the
signal through multiple RISs is weak, we ignore the multi-hop
transmission among multiple RISs in \eqref{ris-eq} \cite{RIS1}.
The received SNR of device $n$ at pAP is:
\begin{align} 
g^r_{n}&=\dfrac{ \left\vert (h^{d}_{n})^* + \sum_{k=1}^{K} \mathbf{h}^H_{n,k} \boldsymbol{\Phi}_k \mathbf{h}^r_k  
\right\vert^2 P_n}{\sigma_0},
\end{align}
As a result, the corresponding rate of device $n$ is:
\begin{align}
r_{n} &=W \mathrm{log}_2 \left( 1+ g^r_{n} \right). 
\end{align}

\subsection{Channel Estimation in RIS Scenario}\label{RIS Ch-est}
We assume that all the channels are subject
to quasi-static fading and hence the channel coefficients remain
constant within one channel coherence interval.
A key property for channel estimation in RIS-aided systems is that there is a scaling ambiguity issue that prevents the RIS-pAP channel and the device-RIS channel from being individually identifiable\cite{RIS-est}.

By rewriting the cascaded channel in \eqref{ris-eq} we have: $\mathbf{h}^H_{n,k} \boldsymbol{\Phi}_k \mathbf{h}^r_k = \boldsymbol{\phi}_k^T \mathbf{u}_{n,k}$ where $\mathbf{u}_{n,k}= \mathrm{diag}(\mathbf{h}^H_{n,k})\mathbf{h}^r_k \in \mathbb{C}^{J_k} $ and $\boldsymbol{\phi}_k=[e^{j\phi_{k1}},\cdots,e^{j\phi_{kJ_k}}]^T$.
Hence, as it has been shown in \cite{RIS-est}, the cascaded channel $\mathbf{u}_{n,k}$
and the direct channel $h^{d}_{n}$ are sufficient for designing RIS-aided
communications. As a result, most of the existing contributions
are focused on designing algorithms to estimate the cascaded channels and the direct channel $h^{d}_{n}$ separately \cite{RIS-est}.

Since a TDMA scheme is utilized, channel estimation is performed separately for each device.
Furthermore, by utilizing a method similar to the discrete Fourier transform (DFT) training scheme proposed in \cite{EG2}, we estimate the cascaded channel for each RIS individually. Specifically, drawing inspiration from the on-off scheme used in \cite{on-off}, we activate the $k$-th RIS while deactivating all others.
Fig. \ref{RIS_training} illustrates the process of channel estimation over time. Under this assumption, we rewrite \eqref{ris-eq} to estimate the channel for the $k$-th RIS as follows:
\begin{align}
y_{n,k} &= \sqrt{P_n} [1, \boldsymbol{\phi}_k^T]
  \Tilde{\mathbf{u}}_{n,k} x_n+ z_0, \ \forall  n \in \mathcal{N},
\end{align}
where $\Tilde{\mathbf{u}}_{n,k} = [ h^{d}_{n},\mathbf{u}_{n,k}]$.
Our aim is to estimate $\Tilde{\mathbf{u}}_{n,k}$ based on the pilot vector $\boldsymbol{\phi}_k$ that is
assumed to be known and depends on the pilot sequence $L_k$ being
used for channel estimation. To ensure that $\Tilde{\mathbf{u}}_{n,k}$ can be uniquely
estimated,  the number
of time slots for channel training of $k$-th RIS, need to satisfy the condition
$L_k \geq (J_k + 1)$.
In general, based on \cite{RIS-est}, there are two common methods for estimation: 1) Least Squares (LS) Estimator and 2) Linear Minimum Mean-Squared-Error (LMMSE).
The LMMSE method has a smaller estimation error as compared to the LS method since prior knowledge of the
channel distributions is used in this approach.
Based on LMMSE, the error variance of the channels when all the channels undergo uncorrelated
Rician fading, is given by \cite{RIS-est}:
\begin{align}\label{LMMSE-e}
\mathrm{var}(\hat{h}_{n}^d) = \mathrm{var}(\hat{\boldsymbol{u}}_{n,k}) = \dfrac{1}{1+L_k\cdot P_n/\sigma_0}, \ \forall k \in \mathcal{K}
\end{align}
where $\sigma_0$ is the noise power. Interestingly, the error variance obtained from LMMSE method in equation \eqref{LMMSE-e} is the same as \eqref{MMSE} in the case of multiple APs.
Therefore, in special case that $J_k=J$ and $L_k =L, \ \forall k \in \mathcal{K}$ the total training sequence time for device $n$ as shown in Fig. \ref{RIS_training} is obtained by:
\begin{align}
T_1 =  K \cdot L \cdot T_s,
\end{align}
which requires a large amount of pilot training overhead, by increasing $K$ and $J_K$ compared to relay.
In typical setups, this results in an excessive training overhead, e.g., when the number of reflecting elements $J_k$ is large. In this case, the remaining time slots for data transmission are, in fact, significantly reduced. \footnote{To reduce the pilot overhead, the element grouping (EG) method is used in \cite{EG1,EG2}. In practical RIS-aided communication systems, the EG method groups the adjacent elements and assigns the same reflection pattern to them. This scheme is effective when the reflecting elements are installed closely together.}

For the rest of discussion, we assume the LMMSE estimators with $L_k = (J_k + 1)$.

\section{Proposed Method for Transmit Power Optimization with Multiple APs}\label{opt}
In IIoT subnetworks, minimizing the overall emitted power extends device battery time, and reduces interference in neighbor subnetworks. In this section, we propose a method for minimization of the total transmit power of all devices and APs, while coping with the delay constraints presented in the previous sections.
The optimization problem can be posed as
\begin{subequations} \label{P1}
	\allowdisplaybreaks
	\begin{align} &\min_{\mathbf{P},\mathcal{N}_{1h},\mathcal{N}_{2h},\mathcal{D}}\left( \sum_{n \in \mathcal{N}} P_n + \sum_{n \in \mathcal{N}_{2h}}P_{D_n}^s \right), \label{P1-a}\\
	\mathrm{s.t.}\ & 
T^{UL}\leq  T', \label{P1-b}
\\&P_n \leq P_{\max}, \ \forall n \in \mathcal{N}\label{P1-c}\\
	& 
	P_{D_n}^s \leq P_{\max}, \ \forall n \in \mathcal{N}_{2h}\label{P1-d}	
	\end{align} 
\end{subequations}
where $\mathbf{P} = [P_1, P_2, \ldots, P_N, P_{D_1}^s, P_{D_2}^s, \ldots, P_{D_{N_{2h}}}^s]$ represents the vector of transmit powers for all devices, with $P_{\max}$ denoting the maximum permissible transmission power. Constraint \eqref{P1-b} addresses the requirement for low latency, while \eqref{P1-c} and \eqref{P1-d} set the power limits.

\subsection{Solution in DF Relaying} \label{solution1}
\subsubsection{TDMA Case}
Given the complexity of jointly selecting relays and minimizing power in \eqref{P1}, finding the optimal solution directly is challenging. Therefore, we employ a sequential method that first identifies the optimal transmission link using a fixed power allocation and then subsequently minimizes the transmit power.
 By assuming constant power, we reformulate the objective to identify the link that maximizes transmission rate, thereby reducing total delay. Following \cite{Laneman}, the average rate in DF relaying is given by $1/2 \cdot \min \left\{ r_{n,D_n}^{(1)}, r_n^{(2)} \right\}$. Given $P_n$ and $P_{D_n}^s$, we compute the minimum channel gain over both transmission phases for all potential relays, selecting the link with the highest gain, i.e., we have:
\begin{align}\label{relay-select}
&g^{2h}_{n,k} \triangleq \frac{1}{2}\min \left\{|h_k^a|^2, |h^s_{n,k}|^2\right\},  \ \forall n \in \mathcal{N}, \forall k \in \mathcal{K} \\
&g^{2h}_n \triangleq \max[g^{2h}_{n,1},g^{2h}_{n,2},\cdots,g^{2h}_{n,K}], \ \forall n \in \mathcal{N}
\label{16}
\end{align}
This metric is subsequently compared with the direct channel gain $|h^d_n|^2$ to ascertain whether direct or cooperative transmission is more advantageous.
Should two-hop transmission prove preferable, the sAP exhibiting the maximum channel gain will be chosen for cooperative transmission. The detailed procedure is presented in Algorithm \ref{Alg1}.
\begin{algorithm}
	\caption{Algorithm for classification of devices and relay selection in DF case}
	\label{Alg1}
	\begin{algorithmic}[1]
		\State \textbf{Input:}
		The channel gains $h_k^a$, $h^s_{n,k}$, and $h^d_n$    for all $N$ devices and $K$ sAPs.
		\For {$n= 1: N$} 
		\State calculate $g^{2h}_n$ from  \eqref{16}
            \If{$|h^d_n|^2 \geq g^{2h}_n$}
                \State $n \rightarrow \mathcal{N}_{1h}$ 
             \Else
                 \State $n \rightarrow \mathcal{N}_{2h}$
                 \State select relay: $ \arg \displaystyle\max_k \{g^{2h}_{n,k}\} \rightarrow {D}_{n}$
            \EndIf
		\EndFor
		\State Output: $\mathcal{N}_{1h}$, $\mathcal{N}_{2h}$, $\mathcal{D}$.
	\end{algorithmic}
\end{algorithm}

Given the sets $\mathcal{N}_{1h}$, $\mathcal{N}_{2h}$, and $\mathcal{D}$, we proceed to minimize transmit power in \eqref{P1}.
In this case we sustitude $T^{UL}$ in \eqref{P1-b} as follows:
\begin{align}\label{P2-b}
 T^{UL} &=
\sum_{n \in \mathcal{N}_{1h}} \dfrac{B_n}{W\theta	 \mathrm{log}_2 \left( 1+ \hat{g}_{n}^d  \right)} \\ \nonumber & + \sum_{n \in \mathcal{N}_{2h}} \dfrac{B_n}{W \theta \displaystyle \mathrm{log}_2 \left( 1+ \hat{g}_{n,D_n}  \right) }\\ \nonumber& +\sum_{n \in \mathcal{N}_{2h}} 
\dfrac{B_n}{W \theta \mathrm{log}_2 \left( 1+  \hat{g}_{D_n,p}  + \hat{g}_n^d\right)}   
\end{align}
It is important to note that problem \eqref{P1} with the surrogate $T^{UL}$ from \eqref{P2-b} is formulated for the general I-CSI case. For the P-CSI case, $\theta = 1$ and $T' = T$.
Problem \eqref{P1} is hard to solve because of the non-convex constraint \eqref{P1-b}, and thus finding the global optimum is generally intractable. To circumvent the non-convexity, we resort to SPCA where the problem is iteratively approximated by a sequence of convex programs. At each iteration, the non-convex constraint is replaced by convex surrogate that serves as approximation.
Thus, let us rewrite the \eqref{P1} as:
\begin{subequations} \label{P3}
	\allowdisplaybreaks
	\begin{align} 
	&\min_{\mathbf{P}}\left( \sum_{n \in \mathcal{N}} P_n + \sum_{n \in \mathcal{N}_{2h}}P_{D_n}^s \right), \label{P3-a}\\
	\mathrm{s.t.}\ & 
\sum_{n \in \mathcal{N}_{1h}} \dfrac{B_n}{\gamma_n^d} +\sum_{n \in \mathcal{N}_{2h}}  B_n \left ( \dfrac{1}{\gamma_{n,D_n}^{(1)}}+\dfrac{1}{ \gamma_{n}^{(2)}} \right)
\leq T' W\theta,\label{P3-b} 
\\&
\mathrm{log}_2 \left( 1+ \hat{g}_{n}^d  \right) \geq \gamma_{n}^d, \ \ \forall  n \in \mathcal{N}_{1h},\label{P3-c}
\\&
\mathrm{log}_2 \left( 1+ \hat{g}_{n,D_n}  \right) \geq \gamma_{n,D_n}^{(1)}, \  \forall  n \in \mathcal{N}_{2h},   \label{P3-d}
\\&
\mathrm{log}_2 \left( 1+  \hat{g}_{D_n,p}  + \hat{g}_{n}^d \right) \geq \gamma_{n}^{(2)}, \   \forall  n \in \mathcal{N}_{2h},\label{P3-e}
\\&
\eqref{P1-c}-\eqref{P1-d}	
	\end{align} 
\end{subequations}
where $\gamma_n^d$, $\gamma_{n,D_n}^{(1)}$, and $\gamma_{n}^{(2)}$ are auxiliary variables to approximate the non-convex terms with convex counterparts. 
It can be perceived that $\gamma_n^d$, $\gamma_{n,D_n}^{(1)}$, and $\gamma_{n}^{(2)}$ play the roles of lower bound for $\mathrm{log}_2 \left( 1+ \hat{g}_{n}^d  \right)$, $\mathrm{log}_2 \left( 1+ \hat{g}_{n,D_n}  \right)$, and $\mathrm{log}_2 \left( 1+ \hat{g}_{D_n,p}  + \hat{g}_{n}^d  \right)$, respectively. 
Increasing the lower-bound values and simultaneously
reducing the upper-bounds will boost the left-side of
the constraints, which is needed here, so that the constraints
\eqref{P3-b}-\eqref{P3-e} would be active at the optimum.
The \eqref{P3-b} is convex, since it is a linear combination of three quadratic terms over linear functions that is convex \cite{CVX}. 
Affine approximations of constraints \eqref{P3-c}-\eqref{P3-e}, $\forall n \in \mathcal{N}$ are given by:
\begin{subequations}\label{lem1-formul}
	\allowdisplaybreaks
         \begin{align}
	&1 + \rho_{n} - 2^ {\gamma_n^d}  \geq 0,\ \forall  n \in \mathcal{N}_{1h},
	\\&
 \rho_{n} \leq \dfrac{P_n|h^d_n|^2}{\sigma_0},\ \forall  n \in \mathcal{N}_{1h},
	\\&
	1 + \psi_{n} - 2^ {\gamma_{n,D_n}^{(1)} }  \geq 0,\ \forall  n \in \mathcal{N}_{2h}, 
	\\&
	 \psi_{n} \leq \dfrac{P_{n}
\displaystyle |h^s_{n,D_n} |^2}{\sigma_0},\ \forall  n \in \mathcal{N}_{2h} \\&
1 + \zeta_{n} - 2^ {\gamma_{n}^{(2)} }  \geq 0,\ \forall  n \in \mathcal{N}_{2h}, 
	\\&
	\zeta_{n} \leq  \dfrac{P_{D_n}^s|  h^a_{D_n}|^2
 +P_{n}|h^d_n|^2}{\sigma_0},  \forall n \in \mathcal{N}_{2h},
	\end{align} 
\end{subequations}
where $\rho_{n}$, $\psi_{n}$, and $\zeta_{n}$, are auxiliary variables.

Thus, by replacing constraints \eqref{P3-c}--\eqref{P3-e} with \eqref{lem1-formul}, the optimization problem \eqref{P3} transforms into a standard convex semidefinite programming (SDP). This can be efficiently solved using numerical solvers, such as the SDP tool in CVX \cite{CVX}.
\subsubsection{FDMA Case}
 Given the complexity of jointly selecting relays and minimizing power in \eqref{P1}, our approach first identifies the optimal transmission link before minimizing transmit power.
By assuming constant power and uniform bandwidth allocation, we reformulate the objective to identify the link that maximizes transmission rate, thereby reducing total delay. It is equivalent to compute the maximum channel gain for all sAPs and pAP, and selecting the link with the highest gain, based on \eqref{relay-select}. Thus, Algorithm \eqref{Alg1} is used to select the strongest link  for transmission.

Given the sets $\mathcal{N}_{1h}$, $\mathcal{N}_{2h}$, and $\mathcal{D}$, we proceed to minimize transmit power in \eqref{P1}. This is achieved by simplifying the optimization problem outlined in \eqref{P1} as follows:
\begin{subequations} \label{36}
	\allowdisplaybreaks
	\begin{align} &\min_{\mathbf{P},\boldsymbol{\beta},{\alpha}}\left( \sum_{n \in \mathcal{N}} P_n + \sum_{k \in \mathcal{K}}P_{k}^s 
 \right), \label{36-a}\\
	\mathrm{s.t.}\ & 
 \dfrac{B_n}{\beta_{n} W\theta	 \mathrm{log}_2 \left( 1+ \dfrac{P_n|h^d_n|^2}{\beta_{n} \sigma_0}  \right)} \leq 
T',  \ \forall n \in \mathcal{N}_{1h} \label{36-b}
\\&
 \dfrac{B_n}{\beta_{n} W\theta	 \mathrm{log}_2 \left( 1+ \dfrac{P_n|h^s_{n,D_n}|^2}{\beta_{n} \sigma_0}  \right)}  \leq 
\alpha T' \label{36-c}
 \\ &
   \dfrac{B_{n}}{\beta^s_{n} W\theta	 \mathrm{log}_2 \left( 1+ \dfrac{P_{D_n}^s|h^a_{D_n}|^2}{\beta^s_{n} \sigma_0}  \right)} \label{36-d}
   \\ \nonumber&
 \leq (1-\alpha-\alpha^p)
T',  \ \forall n \in \mathcal{N}_{2h} 
\\&
\sum_{n \in \mathcal{N}} \beta_{n} = 1
\label{36e}
\\&
\sum_{n \in \mathcal{N}_{2h}}(\beta^s_{n} -\beta_{n})\leq  0
\label{36f}
\\&
 0 \leq \alpha  \leq 1
\\&
\eqref{P1-c}-\eqref{P1-d}	
	\end{align} 
\end{subequations}
where $\boldsymbol{\beta} = [\beta_{1},\beta_{2},\cdots,\beta_{N},\beta^s_{1},\cdots,\beta^s_{K}]$.
Due to coupling of the variables and nonconvexity of constraint \eqref{36-b}-\eqref{36-d}, the global optimum solution is not straightforward. Thus, we resort to alternating optimization  for finding the suboptimal solution. 
First, by setting constant values for transmit power we allocate the bandwidth to devices so as  the minimum rate of devices is maximized. Since, by maximizing the minimum rate, the delay constraint for each device is automatically satisfied. 
\begin{subequations} \label{37}
	\allowdisplaybreaks
	\begin{align} &\max_{\boldsymbol{\beta}} \left[\min_n
		\left( r_n^d, \dfrac{1}{2} \min \left\{ r_{n,D_n}^{(1)}, r_n^{(2)} \right\} \right) \right ], \label{P2-a}\\
		\mathrm{s.t.}\ & 
		\eqref{36e},\eqref{36f}.
	\end{align} 
\end{subequations}
The OF is not convex. Therefore, we reformulate the problem as follows:
\begin{subequations} \label{38}
	\allowdisplaybreaks
	\begin{align} &\max_{\boldsymbol{\beta}} \ r_{\min}, \label{P2-a}\\
		\mathrm{s.t.}\ &  
		r_{\min} \leq \beta_{n} 	 \mathrm{log}_2 \left( 1+ \dfrac{P_n|h^d_n|^2}{\beta_{n} \sigma_0}  \right),  \ \forall n \in \mathcal{N}_{1h}
			\\&
			r_{\min} \leq 
   \dfrac{\beta_{n}}{2}  \mathrm{log}_2 \left( 1+ \dfrac{P_n|h^s_{n,D_n}|^2}{\beta_{n} \sigma_0}  \right),  \ \forall n \in \mathcal{N}_{2h}
			\\&
			r_{\min} \leq 
   \dfrac{\beta^s_{n}}{2}
   \mathrm{log}_2 \left( 1+ \dfrac{P_{D_n}^s|h^a_{D_n}|^2}{\beta^s_{n} \sigma_0}  \right),\ \forall n \in \mathcal{N}_{2h}
		\\&
				\eqref{36e},\eqref{36f},
	\end{align} 
\end{subequations}
where $r_{\min}$ is a lower band for the OF in \eqref{37}. 
Due to the fact
that the OF and all constraints are convex, the
problem \eqref{38} is convex and the optimal solution can be
obtained by using the interior point method.
After obtaining the optimum bandwidth allocation we can minimize transmit power. So, let us rewrite \eqref{36}
as:
\begin{subequations} \label{39}
	\allowdisplaybreaks
	\begin{align} &\min_{\mathbf{P}, \alpha}\left( \sum_{n \in \mathcal{N}} P_n + \sum_{k \in \mathcal{K}}P_{k}^s 
  \right), \label{39-a}\\
	\mathrm{s.t.}\ & 
 \dfrac{B_n}{\gamma_{n}^{d'}} \leq \beta_{n} W\theta
T',  \ \forall n \in \mathcal{N}_{1h} \label{39-b}
\\&
 \dfrac{B_n}{\beta_{n} 	  \gamma_{n}^{ (1)'} } \leq \alpha W\theta T', \ \forall n \in \mathcal{N}_{2h} \\&
    \dfrac{ B_{n}}{\beta^s_{n}	  \gamma_{n}^{ (2)'} }
 \leq (1-\alpha-\alpha^p)
W\theta T', \ \forall n \in \mathcal{N}_{2h} 
\\&1 + \rho'_{n} - 2^ {\gamma_n^{ d'}}  \geq 0,\ \forall  n \in \mathcal{N}_{1h},
	\\&
 \rho'_{n} \leq \dfrac{P_n|h^d_n|^2}{\beta_n \sigma_0},\ \forall  n \in \mathcal{N}_{1h},
	\\&
	1 + \psi'_{n} - 2^ {\gamma_{n}^{(1)'} }  \geq 0,\ \forall  n \in \mathcal{N}_{2h}, 
	\\&
	 \psi'_{n} \leq \dfrac{P_{n}
\displaystyle |h^s_{n,D_n} |^2}{\beta_n \sigma_0},\ \forall  n \in \mathcal{N}_{2h} \\&
1 + \zeta'_{n} - 2^ {\gamma_{n}^{(2)'} }  \geq 0,\ \forall  n \in \mathcal{N}_{2h}, 
	\\&
	\zeta'_{n} \leq  \dfrac{P_{D_n}^s|  h^a_{D_n}|^2}{\beta^s_{n}\sigma_0},  \forall n \in \mathcal{N}_{2h},
 \\&
 0 \leq \alpha  \leq 1
\\&
\eqref{P1-c}-\eqref{P1-d}	
	\end{align} 
\end{subequations}
where $\rho'_{n}$, $\psi'_{n}$, and $\zeta'_{n}$, are auxiliary variables.
Thus, the optimization problem \eqref{39} is an SDP. This can be efficiently solved using numerical solvers, such as the SDP tool in CVX \cite{CVX}.

\subsection{Solution in AF Relaying}\label{solution2}
\subsubsection{TDMA Case}
In this subsection, we aim to minimize the total transmit power of all devices and sAPs in the AF Case. The objective function (OF) in \eqref{P1} applies to AF as well. Therefore, 
the optimization problem \eqref{P1} can be formulated for AF relaying by replacing \eqref{P1-b} with
\begin{align}\label{P5-b}
&\sum_{n \in \mathcal{N}_{1h}} \dfrac{B_n}{ \beta_n \log_2 \left( 1+ \hat{g}_{n}^d \right)} \\ \nonumber 
		& + \sum_{n \in \mathcal{N}_{2h}} \dfrac{2B_n}{\beta_n \log_2 \left(1+ \hat{g}_n^d+  \hat{g}_{n,D_n}^{\mathrm{AF}}\right) } \leq  T' W\theta.  
\end{align}
where $\beta_n$ is a constant in the TDMA scheme. 
By assuming constant power, we reformulate the objective to identify the link that maximizes transmission rate, thereby reducing total delay. Assuming constant and equal power for $P_n$ and $P_{D_n}^s$, we compute $r_n^{AF}$ for all potential relays, selecting the link with the highest rate, i.e., $\forall n \in \mathcal{N}, \forall k \in \mathcal{K} $ we have:
\begin{align}
&r^{\mathrm{AF}}_{n,k} \triangleq \dfrac{1}{2} \mathrm{log}_2 \left(1+ \hat{g}_{n}^d + \hat{g}_{n,k}^{\mathrm{AF}}\right),  \\
&r^{\mathrm{AF}}_{n,D_n} \triangleq \max_k [r^{\mathrm{AF}}_{n,1},r^{\mathrm{AF}}_{n,2},\cdots,r^{\mathrm{AF}}_{n,K}]
\label{33}
\end{align}
This metric is subsequently compared with the direct link rate $r_n^d$ to ascertain whether direct or cooperative transmission is more advantageous.
Should two-hop transmission prove preferable, the sAP exhibiting the maximum channel gain will be chosen for cooperative transmission. The comprehensive procedure is delineated in Algorithm \ref{Alg2}.
\begin{algorithm}
	\caption{Algorithm for classification of devices and relay selection in AF case}
	\label{Alg2}
	\begin{algorithmic}[1]
		\State \textbf{Input:}
		The channel gains $h_k^a$, $h^s_{n,k}$, and $h^d_n$    for all $N$ devices and $K$ sAPs.
		\For {$n= 1: N$} 
		\State calculate $r^{\mathrm{AF}}_{n,D_n}$ from  \eqref{33}
            \If{$r_n^d \geq r^{\mathrm{AF}}_{n,D_n}$}
                \State $n \rightarrow \mathcal{N}_{1h}$ 
             \Else
                 \State $n \rightarrow \mathcal{N}_{2h}$
                 \State selected relay: ${D}_{n}$
            \EndIf
		\EndFor
		\State Output: $\mathcal{N}_{1h}$, $\mathcal{N}_{2h}$, $\mathcal{D}$.
	\end{algorithmic}
\end{algorithm}

Given the sets $\mathcal{N}_{1h}$, $\mathcal{N}_{2h}$, and $\mathcal{D}$, we proceed to minimize transmit power in \eqref{P1}. However, \eqref{P1} encounters nonconvexity due to constraint \eqref{P5-b}. To address this, we resort to the SPCA technique to approximate it with convex counterparts. Thus, let us rewrite \eqref{P1} as
\begin{subequations}\label{P6}
	\allowdisplaybreaks
	\begin{align} &\min_{\mathbf{P}}\left( \sum_{n \in \mathcal{N}} P_n + \sum_{n \in \mathcal{N}_{2h}}P_{D_n}^s \right), \label{P6-a}\\
		\mathrm{s.t.}\ & 
		\sum_{n \in \mathcal{N}_{1h}} \dfrac{B_n}{\gamma_{n}^d}  
		+ \sum_{n \in \mathcal{N}_{2h}} \dfrac{2B_n}{\gamma_{n}^{\mathrm{AF}}} \leq   T' W\theta,\label{P6-b}
		\\
		& \beta_n \mathrm{log}_2 \left( 1+ g_{n}^d  \right) \geq \gamma_{n}^d, \ \ \forall  n \in \mathcal{N}_{1h}\label{P6-c}\\
		& \beta_n \mathrm{log}_2 \left( 1+ \lambda_{n}  \right) \geq \gamma_{n}^{\mathrm{AF}}, \  \forall  n \in \mathcal{N}_{2h}, 
		\label{P6-d}\\
		& \lambda_{n}  \leq  g_{n,D_n}^{\mathrm{AF}} + g_n^d,  \  \forall  n \in \mathcal{N}_{2h},  \label{P6-e}
		\\ & \eqref{P1-c}-\eqref{P1-d}	
	\end{align}
\end{subequations}
where $\gamma_n^d$, $\gamma_{n}^{\mathrm{AF}}$, and $\lambda_{n}$ are auxiliary variables introduced to approximate the nonconvex terms with convex counterparts. These variables $\gamma_n^d$ and $\gamma_{n}^{\mathrm{AF}}$ play the roles of lower bounds for $\mathrm{log}_2 \left( 1+ g_{n}^d  \right)$ and $\mathrm{log}_2 \left( 1+ \lambda_{n}  \right)$, respectively. By increasing the lower-bound values and simultaneously reducing the upper-bounds, the left-hand side of the constraints is enhanced, which is crucial to ensure that constraints \eqref{P6-b}-\eqref{P6-e} are active at the optimum. The convexity of \eqref{P6-b} can be established, as it constitutes a linear combination of two quadratic terms over linear functions, rendering it convex \cite{CVX}.
Affine approximations of constraints \eqref{P6-c}-\eqref{P6-e}, $\forall n \in \mathcal{N}$, $\forall k \in \mathcal{K}$, are given by:
Affine approximations of constraints \eqref{P6-c}-\eqref{P6-e}, $\forall n \in \mathcal{N}$ are given by:
\begin{subequations}\label{lem2-formul}
	\allowdisplaybreaks
         \begin{align}
	&1 + \rho_{n} - 2^{\gamma_n^d/\beta_n} \geq 0,\ \forall n \in \mathcal{N}_{1h},\label{41a}
	\\&
 \rho_{n} \leq \dfrac{P_n|h^d_n|^2}{\beta_n\sigma_0},\ \forall n \in \mathcal{N}_{1h},
\label{41b}	\\&
	1 + \lambda_{n} - 2^{\gamma_{n}^{\mathrm{AF}}/\beta_n} \geq 0,\ \forall n \in \mathcal{N}_{2h}, \label{41c}
\\&
 \lambda_{n} \leq \dfrac{P_{n}P_{D_n} |h^a_{D_n}|^2 |h^s_{n,D_n}|^2}{(P_{n}|h^s_{n,D_n}|^2+P_{D_n}|  h^a_{D_n}|^2 +\beta_{n}\sigma_0)\beta_{n}\sigma_0}\nonumber\\&
+\dfrac{P_{n}|h^d_n|^2}{\beta_{n} \sigma_0},\ \forall n \in \mathcal{N}_{2h}, \label{35d} 
	\end{align} 
\end{subequations}
 where $\rho_{n}$  is auxiliary variable.
 \eqref{35d} is still nonconvex.
 Equation \eqref{Eq2} at the top of this page is the affine approximation of \eqref{35d}, where we have defined
 \begin{figure*}
	\begin{align}\label{Eq2}
	 & \sigma_0 \beta_{n}(\bar{\Theta}^{[i]}(\lambda_{n},P_{n})|h^s_{n,D_n}|^2+\bar{\Theta}^{[i]}(\lambda_{n},P_{D_n}) |h^a_{D_n}|^2 
      +\lambda_{n}\beta_{n} \sigma_0 )\leq
        \Theta^{[i]}(P_{n},P_{D_n})[|h^a_{D_n}|^2 (|h^d_n|^2 +|h^s_{n,D_n}|^2 ) ]
        \nonumber\\& + \Theta^{[i]}(P_{n},P_{n})
        |h^d_n|^2 |h^s_{n,D_n}|^2+P_{n}|h^d_n|^2 \beta_{n}\sigma_0,
        \ \forall n \in \mathcal{N}_{2h}
	\end{align}
	\hrulefill
\end{figure*}
\begin{align}
        &\Theta^{[i]}(x,y) \triangleq \\\nonumber
        & \frac{1}{2} (x^{[i]}+y^{[i]})(x+y)-\frac{1}{4}(x^{[i]}+y^{[i]})^2-\frac{1}{4}(x-y)^2,
    \end{align} 
    and
    \begin{align}
        &\bar{\Theta}^{[i]}(x,y) \triangleq \\\nonumber
        & \frac{1}{4} (x+y)^2 + \frac{1}{4} (x^{[i]}-y^{[i]})^2 - \frac{1}{2}(x^{[i]}-y^{[i]})(x-y)
    \end{align} 
Thus, by replacing constraints \eqref{P6-c}--\eqref{P6-e} with \eqref{lem2-formul} and \eqref{Eq2}, the optimization problem \eqref{P6} transforms into a standard convex semidefinite programming (SDP). This can be efficiently solved using numerical solvers, such as the SDP tool in CVX \cite{CVX}.

\subsubsection{FDMA Case}
By assuming constant power and uniform bandwidth allocation, we reformulate the objective to identify the link that maximizes transmission rate, thereby reducing total delay. It is equivalent to compute the maximum channel gain for all sAPs and pAP, and selecting the link with the highest gain, based on \eqref{33}. Thus, Algorithm \eqref{Alg2} is used to select the strongest link  for transmission.

Given the sets $\mathcal{N}_{1h}$, $\mathcal{N}_{2h}$, and $\mathcal{D}$, we proceed to minimize transmit power. This is achieved by simplifying the optimization problem outlined in \eqref{P1} as follows:
\begin{subequations} \label{45}
	\allowdisplaybreaks
	\begin{align} &\min_{\mathbf{P},\boldsymbol{\beta}}\left( \sum_{n \in \mathcal{N}} P_n + \sum_{k \in \mathcal{K}}P_{k}^s 
 \right), \\
	\mathrm{s.t.}\ & 
 \dfrac{B_n}{\beta_{n} W\theta	 \mathrm{log}_2 \left( 1+ \hat{g}_n^d  \right)} \leq 
T',  \ \forall n \in \mathcal{N}_{1h} \label{45b}
\\&
   \dfrac{2B_n}{\beta_{n} W\theta \log_2 \left(1+ \hat{g}_n^d+  \hat{g}_{n,D_n}^{\mathrm{AF}}\right) } \leq 
T',  \ \forall n \in \mathcal{N}_{2h} 
\\&
\sum_{n \in \mathcal{N}} \beta_{n} = 1
\label{45e}
\\&
\eqref{P1-c}-\eqref{P1-d}	
	\end{align} 
\end{subequations}

Due to coupling of the variables and nonconvexity of constraints, the global optimum solution is not straightforward. Thus, we resort to alternative optimization  for finding the suboptimal solution. 
First, by setting constant values for transmit power we allocate the bandwidth to devices so as  the minimum rate of devices is maximized. Since, by maximizing the minimum rate, the delay constraint for each device is automatically satisfied. 
\begin{subequations} \label{46}
	\allowdisplaybreaks
	\begin{align} &\max_{\boldsymbol{\beta}} \left[\min_n
		\left( r_n^d, r_n^{\mathrm{AF}} \right) \right ], \label{P2-a}\\
		\mathrm{s.t.}\ & 
		\eqref{36e}
	\end{align} 
\end{subequations}
The OF is not convex. Therefore, we reformulate the problem as follows:
\begin{subequations} \label{47}
	\allowdisplaybreaks
	\begin{align} &\max_{\boldsymbol{\beta}} \ r_{\min}, \label{P2-a}\\
		\mathrm{s.t.}\ &  
		r_{\min} \leq \beta_{n} 	 \mathrm{log}_2 \left( 1+ \dfrac{P_n|h^d_n|^2}{\beta_{n} \sigma_0}  \right),  \ \forall n \in \mathcal{N}_{1h}
			\\&
			r_{\min} \leq \dfrac{ \beta_n}{2} \mathrm{log}_2 \left(1+ g_{n}^d + g_{n,D_n}^{\mathrm{AF}}\right),  \ \forall n \in \mathcal{N}_{2h}			
		\\&
				\eqref{36e},
	\end{align} 
\end{subequations}
where $r_{\min}$ is a lower band for the OF in \eqref{46}. 
Due to the fact
that the OF and all constraints are convex, the
problem \eqref{47} is convex and the optimal solution can be
obtained by using the interior point method.
After obtaining the optimum bandwidth allocation we can minimize transmit power.
So, let us rewrite \eqref{45}
as:
\begin{subequations} \label{48}
	\allowdisplaybreaks
	\begin{align} &\min_{\mathbf{P}}\left( \sum_{n \in \mathcal{N}} P_n + \sum_{k \in \mathcal{K}}P_{k}^s 
  \right), \label{48-a}\\
	\mathrm{s.t.}\ & 
 \dfrac{B_n}{\gamma_{n}^{d}} \leq  W\theta
T',  \ \forall n \in \mathcal{N}_{1h} \label{48-b}
\\&
 \dfrac{2B_n}{\gamma_{n}^{\mathrm{AF}}} \leq  W\theta T', \ \forall n \in \mathcal{N}_{2h}
\\&
\eqref{P1-c},\eqref{P1-d},\eqref{41a}-\eqref{41c},\eqref{Eq2}.	
	\end{align} 
\end{subequations}
Thus, the optimization problem \eqref{48} is an SDP. This can be efficiently solved using numerical solvers, such as the SDP tool in CVX \cite{CVX}.

\section{Proposed Method for Transmit Power Optimization with  Multiple RISs}\label{RIS-opt}

In this subsection, we aim to minimize the total transmit power of all devices in the RIS-aided case. 
The optimization problem can be posed as
\begin{subequations} \label{P-RIS1}
	\allowdisplaybreaks
	\begin{align} 
	&\min_{\mathbf{P},\boldsymbol{\phi}} \displaystyle \sum_{n \in \mathcal{N}} P_n\label{P-RIS1-a}\\
	\mathrm{s.t.}\ & 
\sum_{n \in \mathcal{N}} \dfrac{B_n}{W \theta\mathrm{log}_2 \left( 1+ \dfrac{P_n \left \vert \Tilde{h}_n \right\vert^2 }{\sigma_0} \right)} \leq  T', \label{P-RIS1-b}
\\&
\phi_{k,j} \in [0, 2\pi),   \ \forall k \in \mathcal{K},\forall j \in \mathcal{J}_k, \label{P-RIS1-c}\\&
\eqref{P1-c}
       \end{align} 
\end{subequations}
where $\Tilde{h}_n = \left((h^{d}_{n})^* + \sum_{k=1}^{K} \mathbf{h}^H_{n,k} \boldsymbol{\Phi}_k \mathbf{h}^r_k  
\right) $. 
Due to coupling between variables $P_n$ and $\boldsymbol{\phi}$ we resort to a two stage optimization method. First, for constant power assumption we optimize the RIS reflection coefficients then we minimize the power based on the obtained beamforming values. 
\subsection{First Stage}
We first optimize the phase shift vector $\phi$ of problem \eqref{P-RIS1}. Before optimizing $\phi$, we rewrite the combined channel as $\mathbf{h}^H_{n,k} \boldsymbol{\Phi}_k \mathbf{h}^r_k = \boldsymbol{\phi}_k^T \mathbf{u}_{n,k}$.
According to problem \eqref{P-RIS1}, the optimal $\phi$ can be calculated by maximizing the rate, which in turn maximizes the SNR. Under the assumption of constant power, this results in the maximization of the combined channel gain. Therefore, we have:
 \begin{subequations}\label{P4}
	\allowdisplaybreaks
	\begin{align} 
	&\max_{\boldsymbol{\phi}} \displaystyle \left\vert (h^{d}_{n})^* + \sum_{k=1}^{K} \boldsymbol{\phi}_k^T \mathbf{u}_{n,k}
\right\vert^2  \label{P4-a}\\
	\mathrm{s.t.}\ & 
\phi_{k,j} \in [0, 2\pi),   \ \forall k \in \mathcal{K},\forall j \in \mathcal{J}_k, 
       \end{align}
\end{subequations}
Let $\boldsymbol{\phi}_k^*$ be the conjugate vector of $\boldsymbol{\phi}_k$. The total number of
elements for all RISs is denoted by $Q = \sum_{k=1}^{K} J_k$. Denote $\boldsymbol{v} = [\boldsymbol{\phi}_1^*;\cdots;\boldsymbol{\phi}_K^*] \in \mathbb{C}^Q$
 and $\mathbf{u}_{n} = [\mathbf{u}_{n,1};\cdots ; \mathbf{u}_{n,K}] \in
\mathbb{C}^{Q}$. Problem \eqref{P4} can be rewritten as:
\begin{subequations}\label{P8}
	\allowdisplaybreaks
	\begin{align} 
	&\max_{\boldsymbol{v}} \displaystyle \left\vert h^{d}_{n} + \mathbf{u}_{n}^H \boldsymbol{v}
\right\vert^2  \label{P8-a}\\
	\mathrm{s.t.}\ & 
\vert v_{q} \vert =1,   \ \forall q \in \mathcal{Q}, \label{P8-b}
       \end{align}
\end{subequations}
where $\mathcal{Q} = \{1,\dots,Q\}$.
To solve the optimization problem in \eqref{P8}, various methods
were proposed by techniques such as semidefinite relaxation
(SDR) technique \cite{Wu} and successive refinement (SR)
algorithm \cite{Wu2}. However, the SDR method imposes high
complexity to obtain a rank-one solution and the SR algorithm
requires a large number of iterations due to the need for
updating the phase shifts in a one-by-one manner. Now we propose two approaches to solve the problem in \eqref{P8} as follows:

$\bullet$ Approach 1: To handle the nonconvexity of OF, the first-order Taylor series of $\left\vert h^{d}_{n} + \mathbf{u}_{n}^H \boldsymbol{v}
\right\vert^2 $ is used. Then,
we adopt the SPCA method as:
\begin{subequations}\label{P7}
	\allowdisplaybreaks
	\begin{align} 
	&\max_{\boldsymbol{v}} 
 2 \mathcal{R} (( h^{d}_{n} + \mathbf{u}_{n}^H \boldsymbol{v}^{(i-1)})^H \mathbf{u}_{n}^H \boldsymbol{v})+
 \left\vert h^{d}_{n} + \mathbf{u}_{n}^H \boldsymbol{v}^{(i-1)}
\right\vert^2 \nonumber \\&
- 2 \mathcal{R} ((h^{d}_{n} + \mathbf{u}_{n}^H \boldsymbol{v}^{(i-1)})^H \mathbf{u}_{n}^H \boldsymbol{v}^{(i-1)})
 \label{P4-a}\\
	\mathrm{s.t.}\ & 
\vert v_{q} \vert \leq 1,   \ \forall q \in \mathcal{Q}, \label{P7-b}
       \end{align}
\end{subequations}
where the
superscript $(i - 1)$ represents the value of the variable at the
$(i - 1)$-th iteration.
In \cite{RIS-est}, it has been shown that \eqref{P7-b} always holds with
equality for the optimal solution of problem \eqref{P7}. Then, the optimal solution of problem \eqref{P7} is:
\begin{align}
  & \boldsymbol{v} = e^{j\angle(\mathbf{u}_{n}(h^{d}_{n} + \mathbf{u}_{n}^H \boldsymbol{v}^{(i-1)}))}. 
\end{align}
where $\angle(\cdot)$ represents the angle vector of a vector.
The SCA algorithm for solving problem \eqref{P8} is summarized
in Algorithm \ref{Alg3}. The convergence of Algorithm \ref{Alg3} is guaranteed
as shown in Proposition 3 of \cite{Zappone}.
\begin{algorithm}
	\caption{SCA Method for Phase Optimization}
	\label{Alg3}
	\begin{algorithmic}[1]
		\State \textbf{Input:}
		Initialize $\boldsymbol{v}^{(0)}$. Set iteration number $i=1$.
		\State \textbf{repeat} 
		\State   Set $ \boldsymbol{v}^{(i)} = e^{j\angle(\mathbf{u}_{n}(h^{d}_{i} + \mathbf{u}_{n}^H \boldsymbol{v}^{(i-1)}))} $ and $i=i+1$.
        \State
		 \textbf{Until} the objective value \eqref{P7} converges.
		\State Output: $\boldsymbol{\theta} = (\boldsymbol{v}^{(i)})^*$ .
	\end{algorithmic}
\end{algorithm}

$\bullet$ Approach 2:
Instead of using SPCA to solve the non-convex problem in \eqref{P8}, we solve the problem with a closed-form solution by exploiting the special structure of its OF. Specifically, we have the following inequality:
\begin{align}\label{myeq}
    \left\vert h^{d}_{n} + \mathbf{u}_{n}^H \boldsymbol{v}
\right\vert \labelrel\leq{myeq:inequality} \left\vert h^{d}_{n} \right\vert + \left\vert\mathbf{u}_{n}^H \boldsymbol{v}
\right\vert,
\end{align}
where \eqref{myeq:inequality} is due to the triangle inequality and the equality
holds if and only if $\arg (\mathbf{u}_{n}^H \boldsymbol{v})  =  \arg (h^{d}_{n}) \triangleq \phi_0 $
Next,
we show that there always exists a solution $\phi$ that satisfies \eqref{myeq:inequality}
with equality as well as the phase shift constraints in \eqref{P8-b}.
With \eqref{myeq}, problem \eqref{P8} is equivalent to
\begin{subequations}\label{P9}
	\allowdisplaybreaks
	\begin{align} 
	&\max_{\boldsymbol{v}} \displaystyle \left\vert \mathbf{u}_{n}^H \boldsymbol{v}
\right\vert^2  \label{P9-a}\\
	\mathrm{s.t.}\ & 
\vert v_{q} \vert =1,   \ \forall q \in \mathcal{Q}, \label{P9-b}
\\ &
\arg (\mathbf{u}_{n}^H \boldsymbol{v})  = \phi_0.
       \end{align}
\end{subequations}
It is not difficult to show that the optimal solution to
the above problem is given by $\boldsymbol{v} = e^{j(\phi_0 - \arg (\mathbf{u}_{n}^H))}$. Thus, the phase shift of all RIS (angle of $\boldsymbol{v}$) is
given by
\begin{align}\label{phase_cal}
    \arg(\boldsymbol{v}) = \phi_0 - \arg (\mathbf{u}_{n}^H).
\end{align}

\subsection{Second Stage}
After obtaining the optimum phase shift of RISs from stage 1, we can solve the problem \eqref{P-RIS1} and finding the optimum power values for devices.
Using the same approach as the previous sections we can rewrite \eqref{P-RIS1} as:
\begin{subequations} \label{P-RIS2}
	\allowdisplaybreaks
	\begin{align} &\min_{\mathbf{P}} \sum_{n \in \mathcal{N}} P_n , \label{P-RIS2-a}\\
	\mathrm{s.t.}\ & 
\sum_{n \in \mathcal{N}} \dfrac{B_n}{\gamma_n^r} 
\leq T' W\theta
\\
&1 + \rho_{n}^r - 2^ {\gamma_n^r}  \geq 0,\ \forall  n \in \mathcal{N},
	\\&
 \rho^r_{n} \leq \dfrac{P_n|\Tilde{h}_n|^2}{\sigma_0},\ \forall  n \in \mathcal{N},
\\&
\eqref{P1-c}
	\end{align} 
\end{subequations}

It is evident that the optimization problem now takes the form of a standard convex SDP problem, which can be efficiently solved using numerical solvers such as the SDP tool in CVX \cite{CVX}.

\subsection{Convergence and Complexity Analysis} 
In this section, we present a convergence and complexity analysis of the proposed SPCA algorithm. 

\subsubsection{Convergence Analysis}
Given the non-convex nature of the original problem \eqref{P1}, proving convergence to a global minimum is not feasible. However, we can demonstrate convergence to Karush-Kuhn-Tucker (KKT) points under certain regularity conditions. The following lemmas, referenced from \cite{Beck, Hashempour-Bastami}, will be instrumental in the convergence proof. This analysis is applicable to both DF and AF relaying scenarios, so we present it in a general form to avoid repetition. For simplicity, let $\Omega$ represent the feasible set of \eqref{P1}, and let $\Omega^{[i]}$ denote the feasible set at the $i^{th}$ iteration.
\begin{Lemma}
    Let $\mathcal{D} : \mathbb{R}^n \rightarrow \mathbb{R}$ be a strictly convex and differentiable function on a nonempty convex set $S \subseteq \mathbb{R}^n$. Then D is
strongly convex on the set S. 
\end{Lemma}
\begin{proof}
	See \cite{Beck}.
\end{proof}
\begin{Lemma}
    
Let $\{\mathbf{x}^{[i]}\}$ be the sequence generated by the SPCA method.
Then, for every $i \geq 0$: $\mathbf{i})$ $\Omega^{[i]}
\subseteq \Omega$, $\mathbf{ii})$ $\mathbf{x}^{[i]} \in \Omega^{[i]} \cap \Omega^{[i+1]}$, 
$\mathbf{iii})$ $\{\mathbf{x}^{[i]}\}$ is a feasible point of \eqref{P1}, $\mathbf{iv})$, $\sum_n \mathbf{P}_{n}^{[i+1]} \leq \sum_n\mathbf{P}_{n}^{[i]}$. 
\end{Lemma}
\begin{proof}
	See \cite{Beck}.
\end{proof}
\begin{Lemma}
The sequence $\sum_n \mathbf{P}_{n}^{[i]}$ converges.
\end{Lemma}
\begin{proof}
	See \cite{Beck,Hashempour-Bastami}.
\end{proof}
Based on the aforementioned lemmas and adopting the same approach as in \cite{Hashempour-Bastami}, the convergence of our SPCA method for solving the DF relaying problems discussed in Section \ref{solution1} is straightforward.
For the AF case, since \eqref{35d} is relaxed using the first-order approximation in \eqref{Eq2}, the solution to problem \eqref{P6} at iteration $[i]$ remains feasible for iteration $[i + 1]$. Consequently, the objective function is non-increasing with each iteration. Given that the problem is upper-bounded by the transmit power constraints \eqref{P1-c}-\eqref{P1-d}, the proposed algorithm is guaranteed to converge.

\subsubsection{Complexity Analysis}
Algorithms \eqref{Alg1} and \eqref{Alg2} compute the maximum function over $K$ relays for each of the $N$ devices. Therefore, the computational complexity of  each of these algorithms is $\mathcal{O}\left( N K  \right)$.
 The worst-case complexity order of solving the convex problem \eqref{P3},\eqref{37}, \eqref{39}, \eqref{P6} and \eqref{P-RIS2}, by using SPCA method is given
by $ \mathcal{O}\left(l_{\max} N^{1.5} (M_1 N_{1h} + M_2 N_{2h}) \log_2\left(\frac{1}{\epsilon_1}\right)\right)
$ computations \cite{Beck,Hashempour-Bastami}, where $M_1$ and $M_2$ are the constraints applied to $N_{1h}$ and $N_{2h}$ devices, respectively.
$l_{\max}$ is the number of iterations for interior-point method and $\epsilon_1$ is the solution accuracy. Finally, the  overall complexity of phase shift calculation of all RISs from \eqref{phase_cal} is $\mathcal{O}\left(NQ\right)$.
Thus, the overall worst-case complexity of the proposed method is $\mathcal{O}\left(N^{3.5}\right)$, primarily due to the SPCA approach. However, this complexity is manageable and does not compromise the feasibility of our approach, as modern processors can efficiently handle these computations within the pAP, which serves as the main scheduler of the subnetwork.

\section{Simulation Results}\label{Simulat}
In this section, our proposed solution is evaluated via Monte Carlo simulations. 
The main simulation parameters are shown in Table \ref{Table 1} (except in cases where stated otherwise). 
\begin{table}
	\small
	\renewcommand{\arraystretch}{1.3}
	\caption{Simulation parameters setup}
	\centering
	\label{Table 1}
	\resizebox{\columnwidth}{!}{
		\begin{tabular}{|c|p{60mm}|}
			\hline
			Parameter description  &  Value \\
			\hline \hline
			Subnetwork area  &  $3\times 3 \ m^2$ \\		\hline
			Number of devices, $N$ &  $10 \sim 20$  \\		\hline
			Number of pAPs, &  1  \\ 	\hline
			Number of sAPs, $K$ &  Variant based on the scheme  \\ 	\hline
			Data per device & $32 \sim 256$ Bytes  \\ 		\hline
			Cycle duration, $T$ & 0.1 ms \\ 		\hline
	    	Bandwidth & 100 MHz \\ 		\hline
			Carrier frequency &     10 GHz	 \\ 		\hline
			$P_{\max}$ &  $10 \sim 30$ dBm \\ 		\hline
			Power spectral density of the AWGN &  -174 dBm/Hz  \\ 		\hline
      Shadowing standard deviation &  7 dB  \\ 		\hline
		\end{tabular}}
	\end{table}
We consider a $3 \times 3$ $\rm m^2$ subnetwork, serving $10 \sim 20$ devices. This can be the case of a production module in a factory. Devices and APs are uniformly distributed in the subnetwork area. The pAP is expected to receive a $32 \sim 256$~ bytes packet from each device in a total time of $0.1\rm ms$ and bandwidth of $100\rm MHz$. Moreover, the power spectral density of the additive white Gaussian noise is -174 dBm/Hz.
In our simulations, we adopt Rician fading with a K-factor of 7 for all device-to-AP and AP-to-AP links, reflecting the high likelihood of LoS conditions in small subnetworks, as outlined in \cite{3GPP}. The path loss model is based on the 3GPP channel model for indoor factory scenarios \cite{3GPP}. Additionally, in line with \cite{3GPP}, we incorporate a shadow fading model with a standard deviation of 7 dB.

\subsection{Relay Performance}
We consider cooperative schemes with different number of sAPs to be possibly selected, including `1h' and `1 of $K$'. 
First, we evaluate the TDMA scenario using the DF method. The empirical cumulative distribution function (CDF) of transmission power is obtained through 500 iterations of Monte Carlo simulations, where devices and APs are randomly deployed in the subnetwork for each iteration. As shown in Fig.\ref{fig5}, we compare the empirical CDF of transmission power for devices and APs for different number of APs, using a packet size of $B_n=32$ bytes and a maximum transmission power of $P_{\max}=10$ dBm.  Results show a reduction of the transmit power ranging from 2 dB to 4.5 dB for the two-hop schemes with respect to the single-hop scheme. 
This reduction can be attributed to the possibility of selecting relays in advantageous propagation conditions with respect to the link with the pAP.
Clearly, by increasing the number of sAPs, greater power savings can be achieved. However, it is evident that the difference in this power gain diminishes as the value of $K$ becomes large.
\begin{figure} 
	\centering
	\includegraphics[width=0.92\linewidth]{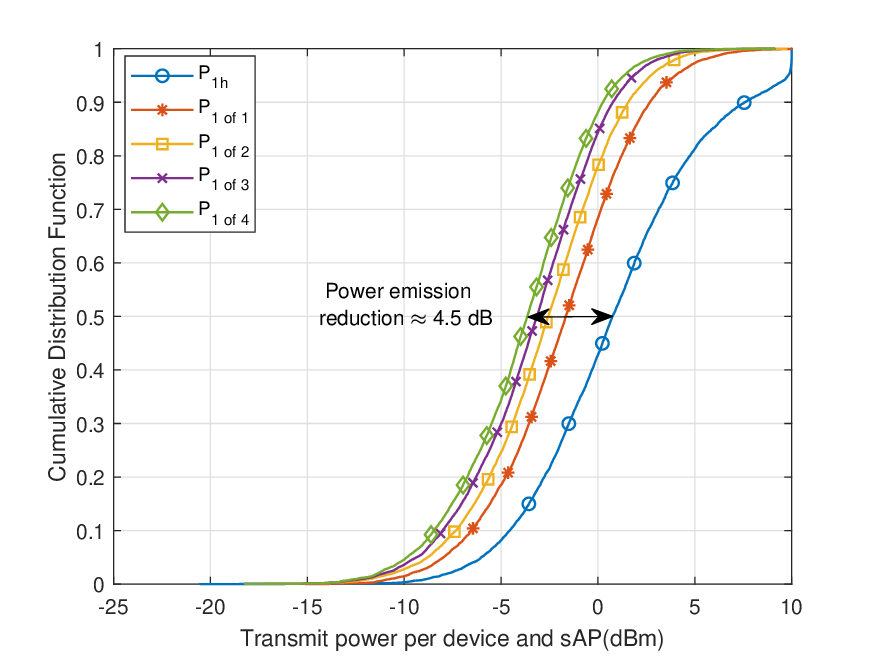}
	\caption{Comparison of the CDF of transmit power across varying numbers of sAPs for TDMA with the DF method, given $B_n=32$ and $P_{\max}=10 $ dBm}
	\label{fig5}		
\end{figure}
In Algorithm \ref{Alg1}, we proposed an efficient method for device classification and relay selection. To evaluate the performance of this method, we compare the transmit power CDFs in Fig.~\ref{fig6} under the following scenarios: 1) All devices employ cooperative transmission, 2) All devices transmit directly to the pAP, 3) Devices are classified for single-hop or two-hop operations according to the classification method in Algorithm~\ref{Alg1}, and  4) Random selection of devices, categorized as single-hop or two-hop devices.

In this simulation, we set $P_{\text{max}} = 10$ dBm, with only one sAP available (i.e., `1 of 1'). The results show that all classification methods require more power compared to Algorithm~\ref{Alg1}, demonstrating its effectiveness in reducing power consumption. Specifically, as shown in Fig.\ref{fig6}, Algorithm \ref{Alg1} outperforms random selection in power emission by up to 3.7 dB.
Since Algorithm~\ref{Alg2} follows a similar procedure, we omit its simulation for the sake of brevity.
\begin{figure} 
	\centering
	\includegraphics[width=0.92\linewidth]{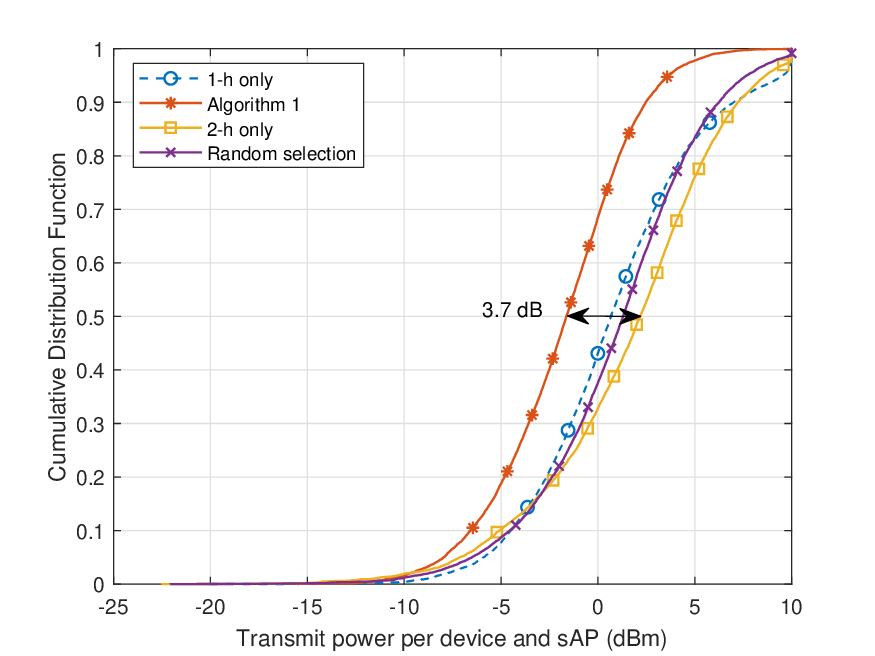}
	\caption{Comparing transmit power CDFs for various device classification methods}
	\label{fig6}		
\end{figure}

In Fig.~\ref{fig7}, we compare the AF and DF relaying methods for both TDMA and FDMA transmissions. Figs. \ref{7a} and \ref{7b} illustrate the results for $B_n = 64$ bytes and $B_n = 128$ bytes, respectively. It is observed that increasing the packet size requires more transmission power to meet the time constraints across all methods. Specifically, the probability of time overflow increases in the TDMA case with $B_n = 128$ bytes. In contrast, the FDMA method outperforms TDMA in terms of power efficiency, as reducing the bandwidth decreases the noise power for each device, thereby reducing the required transmit power. Finally, a comparison between AF and DF shows the superior performance of DF, attributed to its greater flexibility in signal decoding. Notably, for $B_n = 64$ bytes, the DF method achieves power savings of 1.8 dB and 0.8 dB in the TDMA and FDMA cases, respectively, compared to AF relaying.
\begin{figure} 
	\centering
	\subfloat[\label{7a}]{%
		\includegraphics[width=0.92\linewidth]{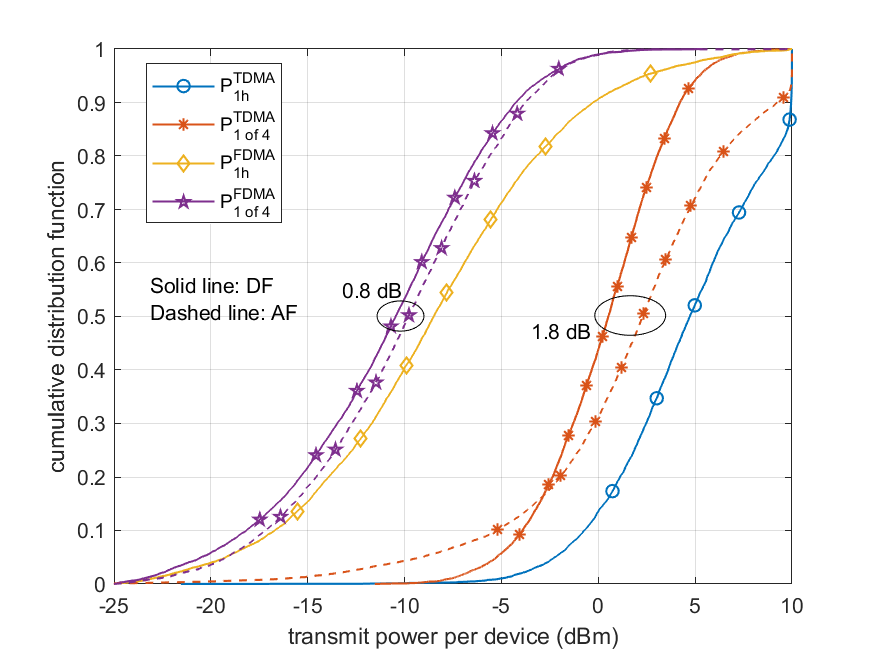}}
	\hfill
	\subfloat[\label{7b}]{%
		\includegraphics[width=0.92\linewidth]{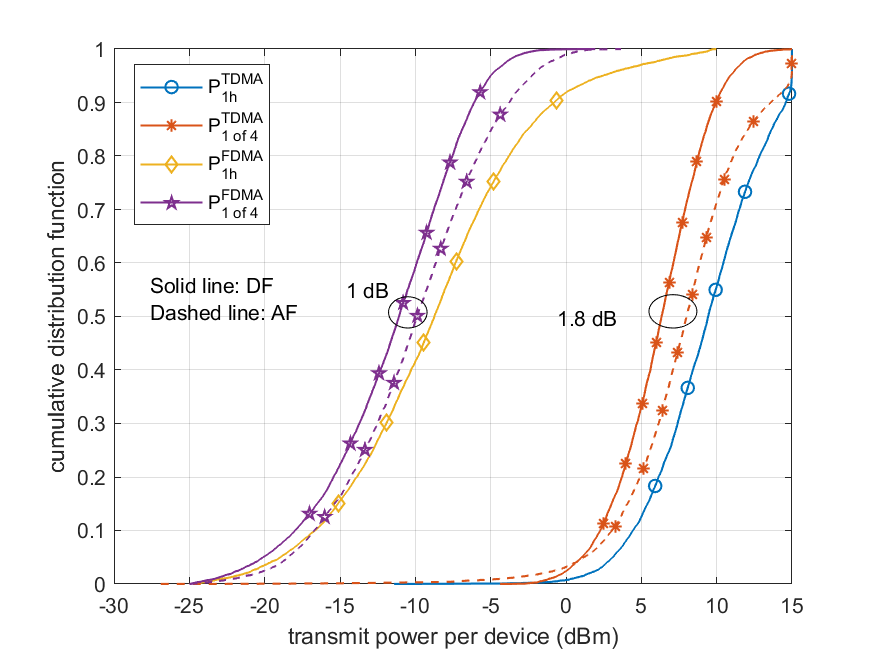}}
	\caption{Comparison of transmit power CDFs for DF and AF methods using TDMA or FDMA transmission with (\protect\subref{7a}) $B_n = 64$ bytes  and (\protect\subref{7b})  $B_n = 128$ bytes.}
		\label{fig7}
\end{figure}

Fig.~\ref{fig9} shows the time overflow rate for different schemes at various power levels in both P-CSI and I-CSI scenarios, with $\theta = 0.5$ and $\theta = 0.9$. These results, derived using maximum power, are independent of the actual power optimization process. Power optimization adjusts power to meet time constraints, defaulting to maximum power in worst-case scenarios, with overflow occurring if this level is insufficient.
In the I-CSI case with $L=4$, less than $1\%$ of time resources are allocated to training, with maximum power used during this phase. The impact of the training phase on average power transmission is minimal. Results for $\theta = 0.5$ and $\theta = 0.9$ in I-CSI are compared with the ideal P-CSI case, showing that a higher $\theta$ reduces the overflow rate, approaching the P-CSI scenario.
For $P_{\mathrm{max}} = 30$ dBm, all two-hop transmissions meet the minimum overflow rate requirement of $10^{-6}$ with $\theta = 0.5$, whereas single-hop transmissions achieve no better than $2 \times 10^{-4}$. In the P-CSI scenario, the `1 of 4' scheme reaches an overflow rate of $10^{-6}$ at $P_{\mathrm{max}} = 5$ dBm, while single-hop rates do not exceed 0.1, highlighting a significant difference.

Fig.~\ref{fig10} shows how the discount factor affects outage probability, with $\theta$ values ranging from 0.5 to 0.9. As $\theta$ approaches 1, the overflow rate decreases because fewer resources are allocated to device packets, but this increase in the discount factor also raises the outage probability. Conversely, a lower $\theta$ increases power consumption to meet stringent latency requirements. To maintain an acceptable overflow rate, as shown in Fig.~\ref{fig9}, the maximum transmit power $P_{\text{max}}$ is set at 25 dBm. The`1 of 3' and `1 of 4' schemes meet the $P_{\text{out}} < 10^{-6}$ constraint for $\theta$ values below 0.7 and 0.8, respectively, while other schemes do not. Notably, a discount factor between 0.7 and 0.8 only requires about a 1 dB increase in power compared to the P-CSI scenario, which is manageable due to the typically favorable propagation conditions in short-range subnetworks.
\begin{figure}
	\centering
	\includegraphics[width=0.92\linewidth]{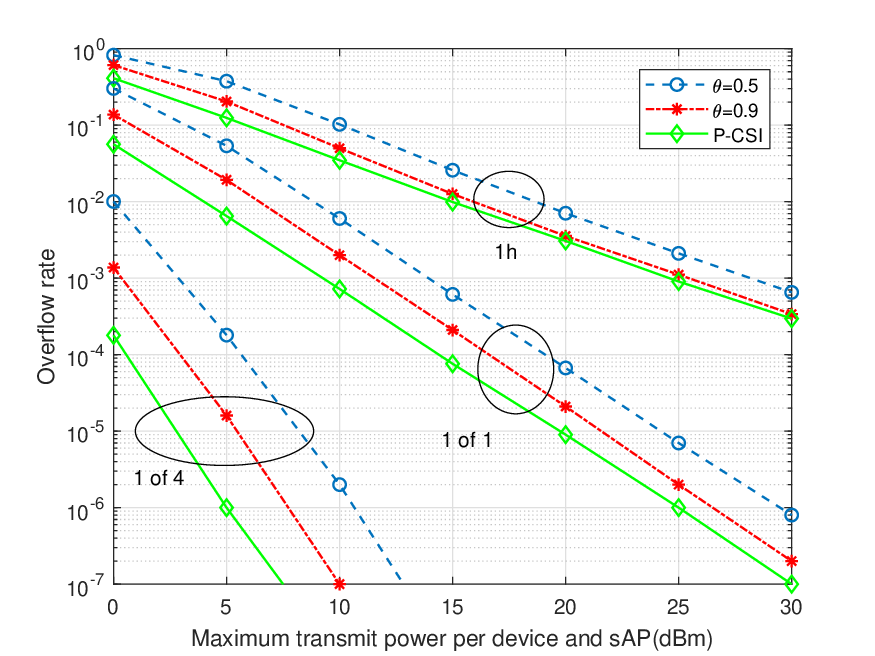}
	\caption{Overflow rate against maximum transmit power of different schemes for P-CSI case and I-CSI case with $L=4$.}.
	\label{fig9}	
\end{figure}
\begin{figure} 
	\centering
	\includegraphics[width=0.92\linewidth]{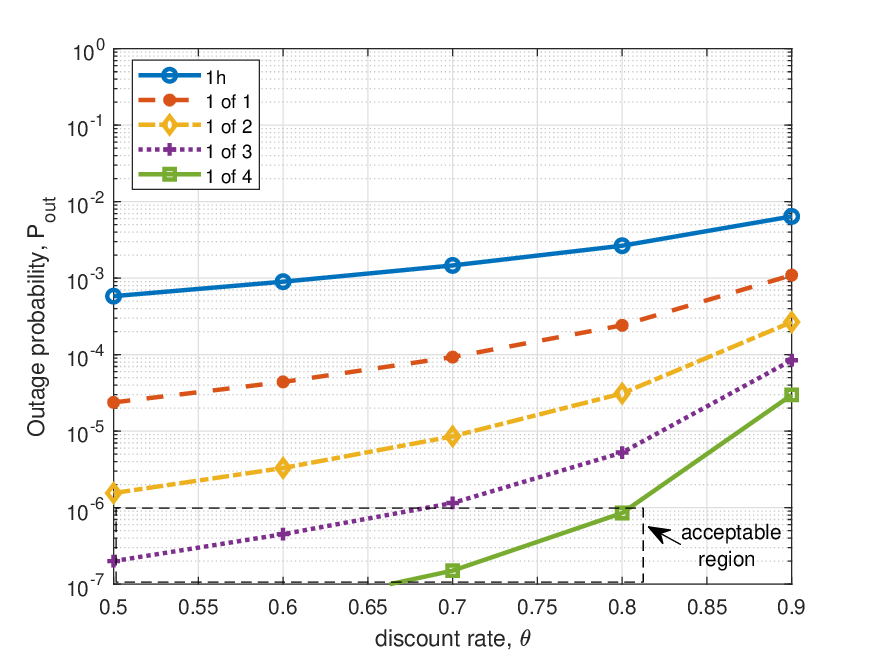}
	\caption{$P_{\text{out}}$ against $\theta$ for I-CSI case with $P_{\max}=25 $ dBm and $L=4$.}
	\label{fig10}	
\end{figure}

\subsection{RIS Performance}
In this subsection, we evaluate the performance of the RIS scenario and compare the results with the cooperative communication method. 

Fig. \ref{fig11} presents the results for RISs with 16 elements. To demonstrate the effectiveness of our approach in setting the phase of the RIS elements, we compare the results with those obtained using random phase settings. It is evident that the performance degrades significantly when random phases are used.
We also compare the deployment of single RIS (1-RIS) and four RISs (4-RIS) schemes. Leveraging four RISs increases the degree of freedom, resulting in a power saving of up to 3 dB. Finally, to illustrate the impact of the discount factor $\theta$ in the case of I-CSI, we present the results for $\theta = 0.5$ and $\theta = 0.9$. It is observed that when $\theta = 0.5$, the data rate decreases by a factor of 2, necessitating additional power to compensate for the loss.
\begin{figure} 
	\centering
	\includegraphics[width=0.92\linewidth]{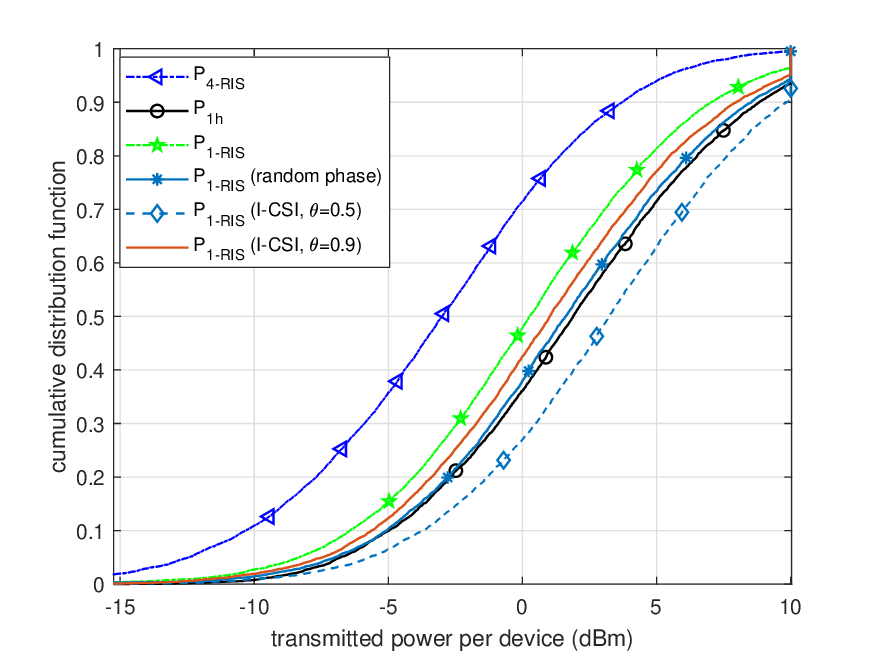}
	\caption{CDF of transmit power for the RIS scenario with different configurations using 16 elements.}
	\label{fig11}	
\end{figure}
Fig. \ref{fig12} illustrates the impact of the number of RIS elements on transmit power. Specifically, deploying 4 RISs with 64 elements significantly reduces transmit power, achieving a reduction of nearly 10 dB compared to single-hop transmission without RIS.
\begin{figure} 
	\centering
	\includegraphics[width=0.92\linewidth]{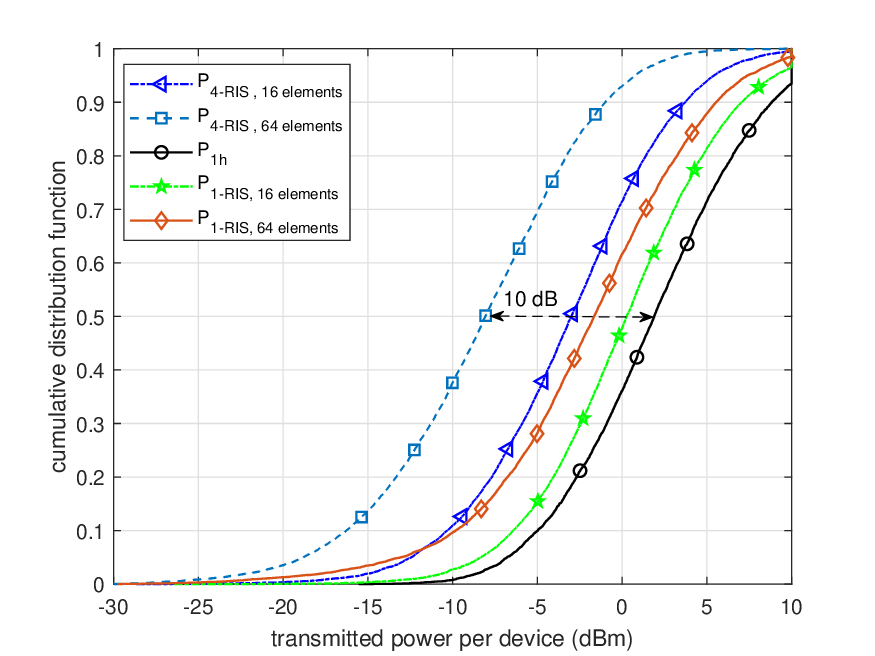}
	\caption{CDF comparison of transmit power for multiple RIS configurations with 16 and 64 elements.}
	\label{fig12}		
\end{figure}

Fig. \ref{fig13} compares the performance of the RIS scenario with relaying. Since RIS can be utilized only in a TDMA setup, we focus on comparing it with the relaying approach under the same conditions. The CDF curve of the `1 of 1' relaying scheme closely matches that of a RIS with 64 elements. However, deploying 4 RISs with 64 elements outperforms the `1 of 4' relaying scheme, achieving up to a 3 dB reduction in transmit power. Thus, Fig. \ref{fig13} highlights the potential of RIS to reduce transmit power compared to  relays.
It is worth noting that our goal is solely to minimize transmit power, which is distinct from minimizing overall power consumption. The power required to operate the RIS is excluded from our current analysis. Therefore, to ensure a fair comparison of power consumption, it is essential to take into account the total power, which will be addressed in future research.
\begin{figure} 
	\centering
	\includegraphics[width=0.92\linewidth]{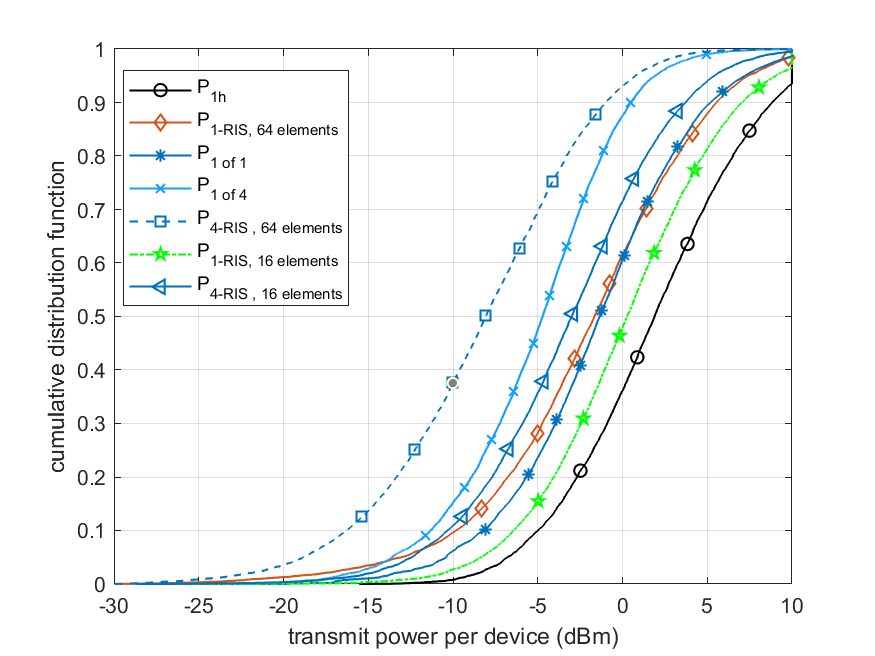}
	\caption{CDF comparison of transmit power for RIS versus relay.}
	\label{fig13}	
\end{figure}
Figs. \ref{fig14} and \ref{fig15} depict the overflow rate and $P_{\text{out}}$ as functions of $\theta$, respectively. It can be observed that relays outperform RIS in terms of reliability. Specifically, using two relays results in both $P_{\text{out}}$ and overflow rates dropping below $10^{-6}$, whereas even with 4 RISs with 16 elements, the overflow rates do not reach $10^{-6}$. This is primarily due to the passive nature of RIS, which cannot match the reliability of an active relay.
\begin{figure} 
	\centering
	\includegraphics[width=0.92\linewidth]{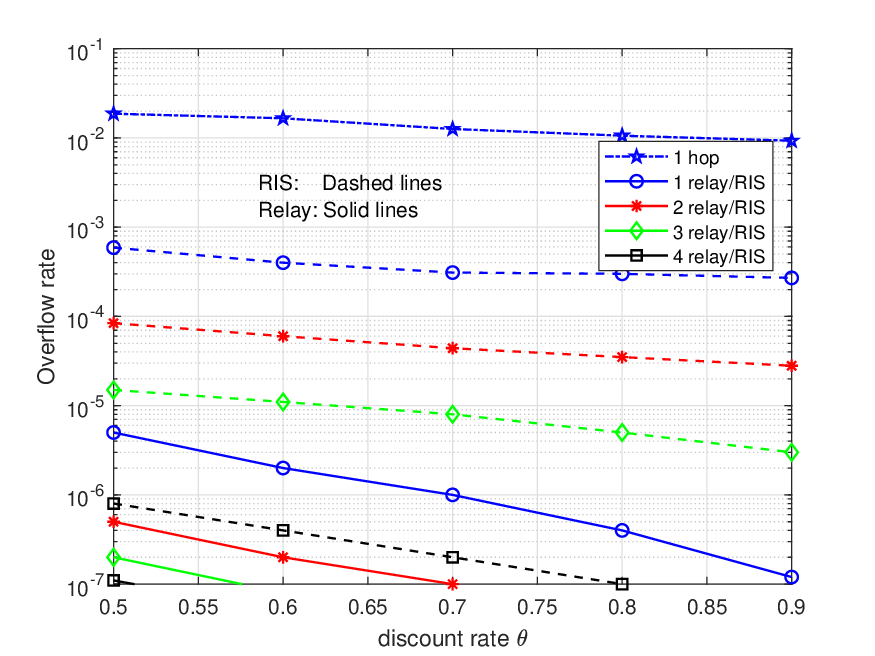}
	\caption{Overflow rate against $\theta$ for I-CSI case with $P_{\max}=23 $ dBm, $J_k=16$ and $L=17$.}
	\label{fig14}	
\end{figure}
\begin{figure} 
	\centering
	\includegraphics[width=0.92\linewidth]{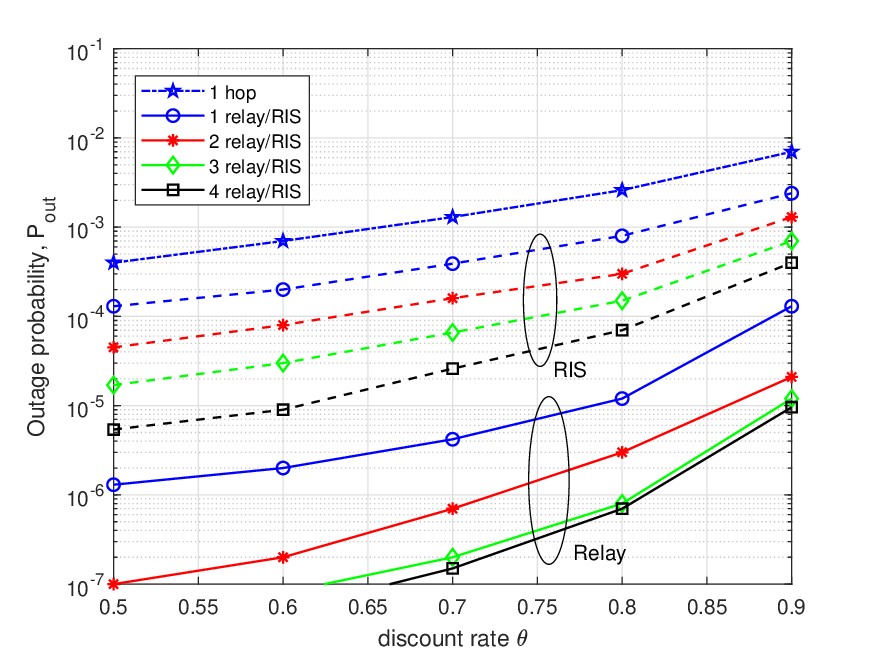}
	\caption{$P_{\text{out}}$ against $\theta$ for I-CSI case with $P_{\max}=23 $ dBm, $J_k=16$ and $L=17$.}
	\label{fig15}	
\end{figure}
By increasing the number of RIS elements and the corresponding training pilots, the results show improvement. In Figs. \ref{fig16} and \ref{fig17}, configurations with 3 and 4 RISs meet the reliability constraint of an outage rate less than $10^{-6}$. However, a single relay can satisfy this constraint with $\theta < 0.7$ using only 65 pilots.

In conclusion, deploying RISs offers a promising approach to reducing transmit power compared to relay-based systems. Specifically, four RISs with 64 elements can achieve up to 3 dB more power savings than an equivalent relay-assisted scenario while also meeting the URLLC requirements of our subnetwork.
However, reducing the number of RISs and/or their elements may compromise the stringent reliability requirements critical for IIoT subnetworks. For instance, four RISs with 16 elements cannot match the reliability of even a single relay in terms of outage probability.
\begin{figure} 
	\centering
	\includegraphics[width=0.92\linewidth]{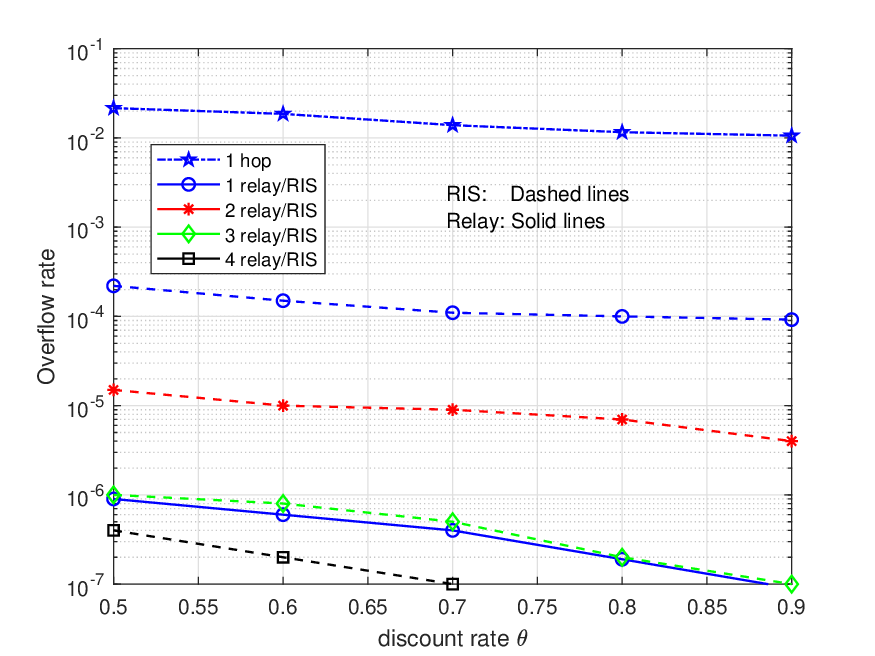}
	\caption{Overflow rate against $\theta$ for I-CSI case with $P_{\max}=23 $ dBm, $J_k=64$ and $L=65$.}
	\label{fig16}	
\end{figure}
\begin{figure} 
	\centering
	\includegraphics[width=0.92\linewidth]{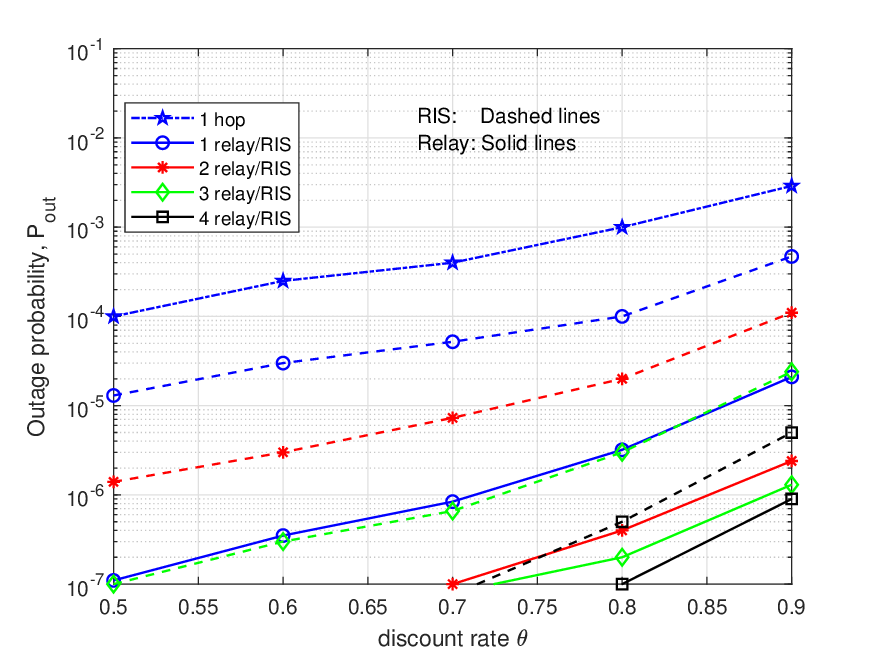}
	\caption{$P_{\text{out}}$ against $\theta$ for I-CSI case with $P_{\max}=23 $ dBm, $J_k=64$ and $L=65$.}
	\label{fig17}	
\end{figure}

\subsection{Discussion on Relay vs. RIS}
In this subsection, we build on the analysis from the previous sections to compare relay and RIS technologies, using both our simulation results and the protocols employed to generate them. Key points of this comparison are summarized as follows:
1) In our simulations, we used the TDMA scheme for both relay and RIS scenarios. While relay systems can operate under both TDMA and FDMA configurations, passive RIS is limited to TDMA due to its lack of frequency selectivity.
2) Our results (specifically in Figs. \ref{fig12} and \ref{fig13}) indicate that RIS significantly reduces transmission power compared to relay systems, owing to its passive nature. Unlike RIS, relays are active devices and inherently consume more power.
3) The DF relay process requires decoding and amplifying signals, which is inherently more complex than simply reflecting signals as RIS does. However, RIS systems require an additional layer of complexity in optimizing each element’s reflection coefficient for each IoT device, adding an optimization step before transmission.
4) Channel estimation is straightforward in the relay scenario, utilizing the LMMSE method. By contrast, channel estimation in the RIS scenario is more challenging due to the cascaded RIS channel, with complexity and overhead increasing as the number of RIS elements grows, as discussed in the channel estimation section.
5) RIS generally leads to lower power emission, which minimizes interference with neighboring subnetworks. Additionally, RIS beamforming directs power toward the intended device, reducing interference with other devices or subnetworks.
6) Our results show that while RIS-assisted approaches provide substantial power savings, relay-based methods offer superior reliability with lower outage probability. This analysis offers a practical basis for selecting the optimal approach based on specific power-saving or reliability requirements.

\section{Conclusions and Future Work}\label{conc}
This paper addressed the challenge of facilitating communication with strict cycle times for numerous devices within an IIoT  subnetwork. We proposed a solution that utilizes multiple sAPs with relaying capabilities and introduced novel TDMA and FDMA communication protocols supporting both single-hop and two-hop device communication. Moreover, we compared the relay-aided scenario with an equivalent RIS-assisted scenario to highlight the strengths and limitations of each approach.
A key focus was on minimizing power consumption while meeting stringent timing constraints.

Simulation results confirm the effectiveness of the proposed protocol. The `1 of 4' scheme reduced power emission by 4.5 dB compared to optimized single-hop transmission. In the P-CSI scenario, the overflow rate decreased from \(9 \times 10^{-4}\) in single-hop to \(10^{-6}\) with the `1 of 1' scheme at 25 dBm, while the `1 of 4' configuration achieved the same \(10^{-6}\) rate with only 5 dBm.
For the RIS scenario, deploying four RISs with 64 elements each reduced transmit power by nearly 10 dB compared to single-hop transmission without RIS. However, relay-based configurations met URLLC targets with just 17 training pilots, whereas four RISs with 16 elements each fell short. Increasing to three RISs with 64 elements matched relay reliability but required significantly more overhead—at least 195 training pilots.
These results underscore a trade-off: RIS-assisted approaches yield substantial power savings, while relay-based methods offer greater reliability, lower outage probability, and reduced overhead. Our findings provide a practical basis for selecting the optimal configuration based on specific needs for power efficiency, reliability, and training overhead.

Future work will extend this research by exploring relaying in IIoT subnetworks with mobile devices, where predicting CSI and accounting for blockage effects will be crucial for efficient resource allocation. The goal will be to further optimize energy efficiency or maximize the sum rate. We will also investigate potential improvements in latency, reliability, and data rate compared to current state-of-the-art solutions. Additionally, research will focus on integrating services with diverse requirements within the cooperative framework, leveraging flexible or full-duplex APs to enhance overall performance.

\begin{IEEEbiography}[{\includegraphics[width=1in,height=1.25in,clip,keepaspectratio]{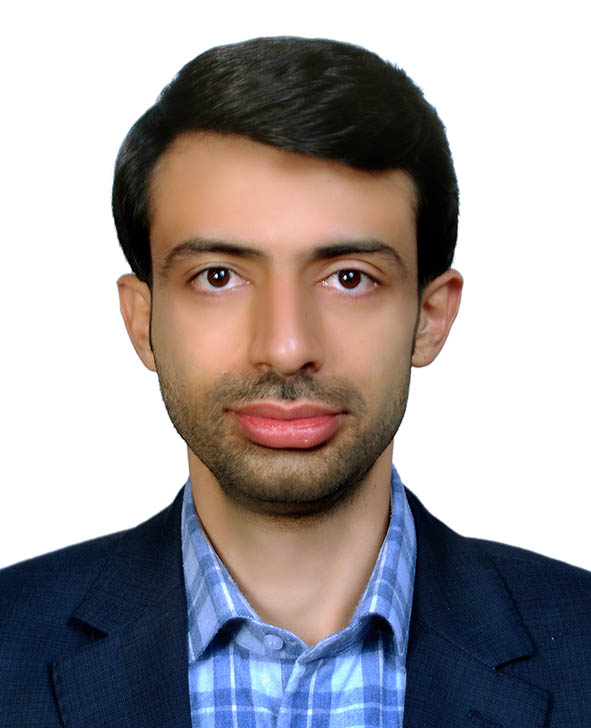}}]{Hamid Reza Hashempour}
	obtained his B.S., M.S., and Ph.D. degrees in Electrical Engineering from Shiraz University, Shiraz, Iran, in 2009, 2011, and 2017, respectively. Currently, he is a Postdoctoral Fellow at the Department of Electronic Systems at Aalborg University, Aalborg, Denmark.
	His research interests encompass wireless communication, physical-layer security, federated learning and 6G in-X subnetworks. In recognition of his work, his paper presented at the ICEE2017 conference was awarded the Best Paper of the Conference.
\end{IEEEbiography}
\begin{IEEEbiography}[{\includegraphics[width=1in,height=1.25in,clip,keepaspectratio]{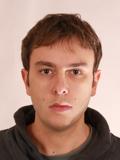}}]{Gilberto Berardinelli}
    received the first and second level degrees (cum laude) in telecommunication engineering from the University of L’Aquila, Italy, in 2003 and 2005, respectively, and the Ph.D. degree from Aalborg University, Denmark, in 2010. He is currently an Associate Professor with the Wireless Communication Networks (WCN) Section, Aalborg University, and external research engineer at Nokia Bell Labs. He has also been involved in multiple European projects, such as FANTASTIC5G, ONE5G, 5GSmartFact, and he is currently coordinator and technical manager of the Horizon Europe 6G-SHINE project, focused on pioneering technology components for short-range wireless communication with extreme requirements. He is author or co-author of more than 150 publications including journals, conferences, book chapters and patents applications. His current research interests are mostly focused on medium access control and radio resource management design for 6G systems and beyond. He is an IEEE Senior Member, and serves as a technical editor for IEEE Wireless Communications magazine.
\end{IEEEbiography}
\begin{IEEEbiography}[{\includegraphics[width=1in,height=1.25in,clip,keepaspectratio]{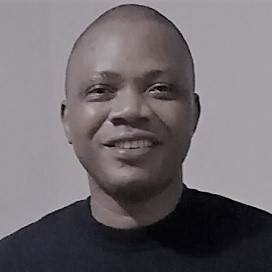}}]{Ramoni Adeogun}
	received the B.Eng. degree in electrical and computer engineering from the Federal University of Technology, Minna, Nigeria, and the Ph.D. degree in electronic and computer systems engineering from the Victoria University of Wellington, New Zealand. He is currently an Associate Professor and leader of the AI for Communications Research Group at Aalborg University, Denmark. In 2021, he was a visiting professor at the Communication Engineering department, Aalto University, Finland.   Prior to joining Aalborg University, he also worked in various positions at the University of Cape Town, South Africa; Odua Telecoms Ltd, Nigeria; Izon Science Ltd, Christchurch, New Zealand; and National Space Research and Development Agency, Nigeria. He has coauthored over 70 publications. His research activities focus on signal processing, machine learning and AI for PHY, MAC and RRM.
\end{IEEEbiography}
\begin{IEEEbiography}[{\includegraphics[width=1in,height=1.25in,clip,keepaspectratio]{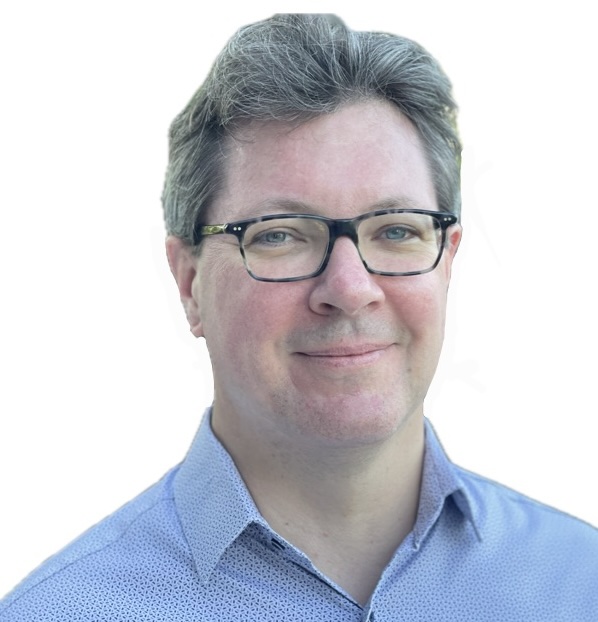}}]{Eduard A. Jorswieck (Fellow, IEEE)}
     received the Ph.D. degree in electrical engineering and computer science from TU Berlin in 2004. From 2006 to 2008, he was with the Signal Processing Group, KTH Stockholm, as a Post-Doctoral Fellow and an Assistant Professor. From 2008 to 2019, he was Professor and Chair for Communication Theory with TU Dresden. He is the Managing Director of the Institute of Communications Technology, the Head of the Chair for Communications Systems, and a Full Professor with Technische Universitat Braunschweig, Brunswick, Germany. He has published more than 190 journal articles, 18 book chapters, one book, four monographs, and some 300 conference papers. His main research interests are in the broad area of communications. He was co-recipient of the IEEE Signal Processing Society Best Paper Award. He and his colleagues were also recipients of the Best Paper and Best Student Paper Awards at the IEEE CAMSAP 2011, IEEE WCSP 2012, IEEE SPAWC 2012, IEEE ICUFN 2018, PETS 2019, ISWCS 2019, and IEEE ICC 2024. Since 2017, he has been the Editor-in-Chief of the EURASIP Journal on Wireless Communications and Networking. He is currently serving on the editorial boards of the IEEE TRANSACTIONS ON INFORMATION THEORY and IEEE TRANSACTIONS ON COMMUNICATIONS. He was on the editorial boards of the IEEE SIGNAL PROCESSING LETTERS, the IEEE TRANSACTIONS ON SIGNAL PROCESSING, the IEEE TRANSACTIONS ON WIRELESS COMMUNICATIONS, and the IEEE TRANSACTIONS ON INFOR- MATION FORENSICS AND SECURITY.
\end{IEEEbiography}


\begin{thebibliography}{99}

\bibitem{Adeogun2020}
R. Adeogun, G. Berardinelli, P. E. Mogensen, I. Rodriguez, and M. Razzaghpour, “Towards 6G in-X subnetworks with sub-millisecond communication cycles and extreme reliability,” \textit{IEEE Access}, vol. 8, pp. 110172–110188, 2020.

\bibitem{Gilberto}
G. Berardinelli et al., “Extreme Communication in 6G: Vision and Challenges for 'in-X' Subnetworks,” \textit{IEEE Open J. Commun. Soc.}, vol. 2, pp. 2516–2535, 2021.

\bibitem{IIoT1}
H. Ren, C. Pan, Y. Deng, M. Elkashlan, and A. Nallanathan, “Joint Pilot and Payload Power Allocation for Massive-MIMO-Enabled URLLC IIoT Networks,” \textit{IEEE J. Sel. Areas Commun.}, vol. 38, no. 5, pp. 816–830, May 2020.

\bibitem{IIoT2}
B. S. Khan, S. Jangsher, A. Ahmed, and A. Al-Dweik, “URLLC and eMBB in 5G Industrial IoT: A Survey,” \textit{IEEE Open J. Commun. Soc.}, vol. 3, pp. 1134–1163, 2022.


\bibitem{Wang2024}
M. Wang, X. Li, Y. Zhang, J. Liu, and A. Smith, ``Improving reliability and throughput in industrial internet of things: Full-duplex relaying, power allocation, and rate adaptation," \textit{IEEE Internet Things J.}, vol. 11, no. 15, pp. 26062–26075, Aug. 1, 2024.

\bibitem{Chen2024}
W. Chen, D. Liu, and Y. Bai, ``B5G/6G URLLC latency reduction method for multisensor industrial internet of things," \textit{IEEE Internet Things J.}, vol. 11, no. 7, pp. 11444–11459, Apr. 1, 2024.


\bibitem{Relay1}
S. Karmakar and M. K. Varanasi, “The Diversity-Multiplexing Tradeoff of the Dynamic Decode-and-Forward Protocol on a MIMO Half-Duplex Relay Channel,” \textit{IEEE Trans. Inf. Theory}, vol. 57, no. 10, pp. 6569–6590, Oct. 2011.

\bibitem{URLLC1}
Y. Hu, M. C. Gursoy, and A. Schmeink, “Relaying-Enabled Ultra-Reliable Low-Latency Communications in 5G,” \textit{IEEE Netw.}, vol. 32, no. 2, pp. 62–68, Mar.–Apr. 2018.

\bibitem{Khosravirad}
S. R. Khosravirad, H. Viswanathan, and W. Yu, “Exploiting Diversity for Ultra-Reliable and Low-Latency Wireless Control,” \textit{IEEE Trans. Wireless Commun.}, vol. 20, no. 1, pp. 316–331, Jan. 2021.

\bibitem{Cheng}
J. Cheng and C. Shen, “Relay-Assisted Uplink Transmission Design of URLLC Packets,” \textit{IEEE Internet Things J.}, vol. 9, no. 19, pp. 18839–18853, Oct. 2022.

\bibitem{Ranjha}
A. Ranjha, G. Kaddoum, and K. Dev, “Facilitating URLLC in UAV-Assisted Relay Systems With Multiple-Mobile Robots for 6G Networks: A Prospective of Agriculture 4.0,” \textit{IEEE Trans. Ind. Informat.}, vol. 18, no. 7, pp. 4954–4965, Jul. 2022.

\bibitem{Joint}
Q. Wen and B.-J. Hu, “Joint Optimal Relay Selection and Power Control for Reliable Broadcast Communication in Platoon,” in \textit{Proc. IEEE 92nd Veh. Technol. Conf. (VTC2020-Fall)}, Victoria, BC, Canada, 2020, pp. 1–6.

\bibitem{Relay-select}
A. G. Onalan, E. D. Salik, and S. Coleri, “Relay Selection, Scheduling, and Power Control in Wireless-Powered Cooperative Communication Networks,” \textit{IEEE Trans. Wireless Commun.}, vol. 19, no. 11, pp. 7181–7195, Nov. 2020.

\bibitem{Power-Alloc}
D. Li, S. R. Khosravirad, T. Tao, P. Baracca, and P. Wen, “Power Allocation for 6G Sub-Networks in Industrial Wireless Control,” in \textit{Proc. IEEE Wireless Commun. Netw. Conf. (WCNC)}, Dubai, United Arab Emirates, 2024, pp. 1–6.

\bibitem{Freq-Alloc}
D. Li, S. R. Khosravirad, T. Tao, and P. Baracca, “Advanced Frequency Resource Allocation for Industrial Wireless Control in 6G Subnetworks,” in \textit{Proc. IEEE Wireless Commun. Netw. Conf. (WCNC)}, Glasgow, U.K., 2023, pp. 1–6.

\bibitem{Mu}
X. Mu et al., “Intelligent reflecting surface enhanced indoor robot path planning: a radio map-based approach,” \textit{IEEE Trans. Wireless Commun.}, vol. 20, no. 7, pp. 4732–4747, Jul. 2021.

\bibitem{Liu}
Z. Liu, Y. Liu, and X. Chu, “Reconfigurable-intelligent-surface-assisted indoor millimeter-wave communications for mobile robots,” \textit{IEEE Internet Things J.}, vol. 11, no. 1, pp. 1548–1557, Jan. 2024.

\bibitem{Zhang}
Y. Zhang, S. R. Khosravirad, X. Chu, and M. A. Uusitalo, “On the IRS Deployment in Smart Factories Considering Blockage Effects: Collocated or Distributed?,” \textit{IEEE Trans. Wireless Commun.}, doi: 10.1109/TWC.2024.3387886.


\bibitem{EuCNC}
H. R. Hashempour, G. Berardinelli, and R. Adeogun, “A Power Efficient Cooperative Communication Protocol for 6G in-Factory Subnetworks,” in \textit{Proc. Joint Eur. Conf. Netw. Commun. \& 6G Summit (EuCNC/6G Summit)}, Antwerp, Belgium, 2024, pp. 670–675.

\bibitem{RIS-or-relay1}
Q. Gu, D. Wu, X. Su, J. Jin, Y. Yuan and J. Wang, ``Performance Comparisons Between Reconfigurable Intelligent Surface and Full/Half-Duplex Relays," in \textit{Proc. IEEE 94th Veh. Technol. Conf. (VTC-Fall)}, Norman, OK, USA, 2021, pp. 1-6.

\bibitem{RIS-or-relay2}
M. Di Renzo \textit{et al.}, ``Reconfigurable Intelligent Surfaces vs. Relaying: Differences, Similarities, and Performance Comparison," \textit{IEEE Open J. Commun. Soc.}, vol. 1, pp. 798-807, 2020.

\bibitem{RIS-or-relay3}
E. Björnson, Ö. Özdogan and E. G. Larsson, ``Intelligent Reflecting Surface Versus Decode-and-Forward: How Large Surfaces Are Needed to Beat Relaying?," \textit{IEEE Wireless Commun. Lett.}, vol. 9, no. 2, pp. 244-248, Feb. 2020.

\bibitem{Jafari}
A. H. Jafari, D. Lopez-Perez, M. Ding, and J. Zhang, “Study on scheduling techniques for ultra-dense small cell networks,” in \textit{Proc. IEEE 82nd Veh. Technol. Conf. (VTC2015-Fall)}, Boston, MA, USA, 2015, pp. 1–6.

\bibitem{3GPP}
3GPP TR 38.901, v17.0.0, “Technical Specification Group Radio Access Network; Study on channel model for frequencies from 0.5 to 100 GHz,” 2022.

\bibitem{AF-ref}
M. R. Souryal and B. R. Vojcic, “Performance of Amplify-and-Forward and Decode-and-Forward Relaying in Rayleigh Fading with Turbo Codes,” in \textit{Proc. IEEE Int. Conf. Acoust., Speech Signal Process. (ICASSP)}, Toulouse, France, 2006, pp. IV–IV.

\bibitem{Multi-RIS}
W. Jiang and H. D. Schotten, “User Selection for Simple Passive Beamforming in Multi-RIS-Aided Multi-User Communications,” in \textit{Proc. IEEE 97th Veh. Technol. Conf. (VTC2023-Spring)}, Florence, Italy, 2023, pp. 1–5.

\bibitem{RIS1}
P. Wang, J. Fang, X. Yuan, Z. Chen, and H. Li, “Intelligent reflecting surface-assisted millimeter wave communications: Joint active and passive precoding design,” \textit{IEEE Trans. Veh. Technol.}, vol. 69, no. 12, pp. 14960–14973, Dec. 2020.

\bibitem{RIS-est}
C. Pan et al., “An Overview of Signal Processing Techniques for RIS/IRS-Aided Wireless Systems,” \textit{IEEE J. Sel. Topics Signal Process.}, vol. 16, no. 5, pp. 883–917, Aug. 2022.

\bibitem{EG2}
B. Zheng and R. Zhang, “Intelligent reflecting surface-enhanced OFDM: Channel estimation and reflection optimization,” \textit{IEEE Wireless Commun. Lett.}, vol. 9, no. 4, pp. 518–522, Jul. 2020.

\bibitem{on-off}
Y. Yang, B. Zheng, S. Zhang, and R. Zhang, “Intelligent reflecting surface meets OFDM: Protocol design and rate maximization,” \textit{IEEE Trans. Commun.}, vol. 68, no. 7, pp. 4522–4535, Jul. 2020.

\bibitem{EG1}
Y. Yang, B. Zheng, S. Zhang, and R. Zhang, “Intelligent reflecting surface meets OFDM: Protocol design and rate maximization,” \textit{IEEE Trans. Commun.}, vol. 68, no. 7, pp. 4522–4535, Jul. 2020.

\bibitem{Laneman}
J. N. Laneman, D. N. C. Tse, and G. W. Wornell, “Cooperative diversity in wireless networks: Efficient protocols and outage behavior,” \textit{IEEE Trans. Inf. Theory}, vol. 50, no. 12, pp. 3062–3080, Dec. 2004.

\bibitem{CVX}
M. Grant and S. Boyd, “CVX: Matlab software for disciplined convex programming, version 2.1,” \textit{http://cvxr.com/cvx}, Mar. 2014.

\bibitem{Wu}
Q. Wu and R. Zhang, “Intelligent reflecting surface enhanced wireless network via joint active and passive beamforming,” \textit{IEEE Trans. Wireless Commun.}, vol. 18, no. 11, pp. 5394–5409, Nov. 2019.

\bibitem{Wu2}
Q. Q. Wu and R. Zhang, “Beamforming optimization for wireless network aided by intelligent reflecting surface with discrete phase shifts,” \textit{IEEE Trans. Commun.}, vol. 68, no. 3, pp. 1838–1851, May 2020.

\bibitem{Zappone}
A. Zappone, E. Björnson, L. Sanguinetti, and E. Jorswieck, “Globally optimal energy-efficient power control and receiver design in wireless networks,” \textit{IEEE Trans. Signal Process.}, vol. 65, no. 11, pp. 2844–2859, Jun. 2017.

\bibitem{Beck}
A. Beck, A. Ben-Tal, and L. Tetruashvili, “A sequential parametric convex approximation method with applications to nonconvex truss topology design problems,” \textit{J. Glob. Optim.}, vol. 47, pp. 29–51, 2010.

\bibitem{Hashempour-Bastami}
H. R. Hashempour et al., “Secure SWIPT in the Multiuser STAR-RIS Aided MISO Rate Splitting Downlink,” \textit{IEEE Trans. Veh. Technol.}, vol. 73, no. 9, pp. 13466-13481, Sept. 2024.



\end{thebibliography}
\end{document}